%% file: EMST-lunar.tex
\documentclass[a4paper,reqno]{amsart}
\usepackage{fullpage}
\usepackage[foot]{amsaddr}
\usepackage{graphicx}
\usepackage{color}
\usepackage{amsthm,amssymb}
\usepackage{xfrac}
\usepackage{subcaption}

\usepackage{tikz}
\usepackage{tikz-cd}
\usepackage[hidelinks]{hyperref}
\usepackage{mathtools}
\usepackage{xcolor}
\usepackage{algpseudocode}
\usepackage{algorithm}
\usepackage{enumerate}
\usepackage{nicematrix}
\usepackage{adjustbox}
\usepackage[normalem]{ulem}

\usepackage{chronology}

\usepackage[colorinlistoftodos,prependcaption,textsize=tiny,textwidth=2cm]{todonotes}
\makeatletter%
\@mparswitchfalse%
\makeatother%
\normalmarginpar%

\newcommand {\mm}[1]   {\ifmmode{#1}\else{\mbox{\(#1\)}}\fi}

\newcommand{\Expect}[2]     {\mm{{\mathbb E}_{#1}\left[{#2}\right]}}
\newcommand{\Prob}[2]       {\mm{{\mathbb P}_{#1}\left[{#2}\right]}}
\newcommand{\Lspace}        {\mm{{\mathbb L}}}
\newcommand{\Rspace}        {\mm{{\mathbb R}}}
\newcommand{\Sspace}        {\mm{{\mathbb S}}}
\newcommand{\Grass}[2]      {\mm{{\mathbb Gr}_{{#1}}^{{#2}}}}
\newcommand{\Volume}[1]     {\mm{{\rm Vol}{({#1})}}}
\newcommand{\One}[1]        {\mm{{\bf 1}_{#1}}}
\newcommand{\Enc}[2]        {\mm{{\rm Enc}{({#1},{#2})}}}
\newcommand{\midpoint}[2]   {\mm{{z}{({#1},{#2})}}}

\newcommand{\domain}[1]     {\mm{{\sf dom}{\left({#1}\right)}}}

\newcommand{\Alpha}[2]      {\mm{{\rm Alf}_{#1}{({#2})}}}

\newcommand{\Voronoi}[2][]  {\mm{{\rm Vor}_{#1}{({#2})}}}
\newcommand{\Delaunay}[2][] {\mm{{\rm Del}_{#1}{({#2})}}}
\newcommand{\Delaunaycol}[2][] {\mm{{\rm Del}^{\rm col}_{#1}{({#2})}}}
\newcommand{\card}[1]       {\mm{{\#}{#1}}}

\newcommand{\affine}[1]     {\mm{\rm aff\,}{#1}}
\newcommand{\diff}          {\mm{\,}{\rm d}}
\newcommand{\interior}[1]   {\mm{\rm int\,}{#1}}
\newcommand{\dom}[1]        {\mm{\rm dom}{({#1})}}

\newcommand{\Edist}[2]      {\mm{\|{#1}-{#2}\|}}
\newcommand{\norm}[1]       {\mm{\|{#1}\|}}
\newcommand{\Cost}[2]       {\mm{{\rm Cost}_{#1}{({#2})}}}
\newcommand{\Costrestr}[2]  {\mm{{\rm Cost}_{#1}^{\circ}{{(#2)}}}}
\newcommand{\Costrooted}[2] {\mm{{\rm Cost}_{#1}^\ast{{({#2})}}}}
\newcommand{\Lune}[2]       {\mm{{\rm Lune}_{#1}{({#2})}}}

\newcommand{\Disk}[2]       {\mm{D}_{#1}{({#2})}}
\newcommand{\aaa}           {\mm{\bf a}}
\newcommand{\bbb}           {\mm{\bf b}}
\newcommand{\ccc}           {\mm{\bf c}}

\newcommand{\ppp}           {\mm{\bf p}}
\newcommand{\uuu}           {\mm{\bf u}}
\newcommand{\vvv}           {\mm{\bf v}}
\newcommand{\xxx}           {\mm{\bf x}}
\newcommand{\VVV}[1]        {\mm{\bf V}_{\!{#1}}}

\newcommand{\Rarc}          {\mm{R_{\rm arc}}}
\newcommand{\Rnode}         {\mm{R_{\rm node}}}
\newcommand{\Rarcprime}     {\mm{R_{\rm arc}'}}
\newcommand{\Rnodeprime}    {\mm{R_{\rm node}'}}

\newcommand{\Skip}[1]       {}

\definecolor{green-blue}{rgb}{0.00, 0.80, 0.30}

\definecolor{blue-red}{rgb}{0.8, 0.00, 0.95}

\newcommand{\ee}            {\mm{\varepsilon}}

\theoremstyle{definition}

\newtheorem{theorem}{Theorem}
\numberwithin{theorem}{section}

\newtheorem{proposition}[theorem]{Proposition}

\newtheorem{lemma}[theorem]{Lemma}
\newtheorem{corollary}[theorem]{Corollary}

\newtheorem*{remark}{Remark}

\newtheorem{definition}[theorem]{Definition}

\numberwithin{equation}{section}

\newtheorem*{lemma*}{Lemma}
\newtheorem*{proposition*}{Proposition}

\title{Lunar Generalizations of the Euclidean Minimum Spanning Tree in the Plane and their Expected Costs}

\author{O.\ Draganov$^{1,6}$}
\email{$^6$ondrej.draganov@inria.fr}
\author{H.\ Edelsbrunner$^{2,7}$}
\email{$^7$edels@ist.ac.at}
\author{S.\ Rosenmeier$^{3,4,8}$}
\email{$^8$rosenmeier@biochem.mpg.de}
\author{M.\ Saghafian$^{2,5,9}$}
\email{$^9$msaghafi@ac.tuwien.ac.at}

\address{$^{1}$INRIA, Sophia Antipolis, France}
\address{$^{2}$ISTA (Institute of Science and Technology Austria), Kloster\-neu\-burg, Austria}
\address{$^{3}$Max Planck Institute of Biochemistry, Munich, Germany}
\address{$^{4}$University of Vienna, Vienna, Austria}
\address{$^{5}$Technical University of Wien, Vienna, Austria}

\thanks{\emph{Funding.} The second author is partially supported by the DFG Collaborative Research Center TRR 109, `Discretization in Geometry and Dynamics', Austrian Science Fund (FWF), grant no.\ I 02979-N35. The first author is supported by `AI4scMed', Agence Nationale de la Recherche, France 2030, grant no.\ ANR-22-PESN-000.}


\begin{document}
\begin{abstract}
  Motivated by the recent introduction of chromatic persistent homology \cite{CDES26}, we generalize the Euclidean minimum spanning tree (EMST) for $n$ points in $\Rspace^2$ to the lunar EMST for the case in which the points come in $s+1$ colors.
  Calling the intersection of $s+1$ disks of radius $r$ centered at points with pairwise different colors a \emph{lune}, the generalized EMST reflects the history of the union of lunes as $r$ goes from $0$ to $\infty$, and its \emph{cost} is twice the difference between the radii when the arcs and nodes of the tree are formed.
  If the points are chosen uniformly at random in $[0,1]^2$ and colored randomly, the expected cost converges to some constant (that depends on $s$) times $\sqrt{n}$, as $n$ goes to infinity.
  The main contribution of this paper is a proof that this constant exists, however similar to the case of the classic EMST, its precise value remains elusive.
\end{abstract}

\maketitle


\section{Introduction}
\label{sec:1}

The Euclidean minimum spanning tree (or EMST) of $n$ points in $\Rspace^2$ reflects the history of the connected components in the union of disks of radius $r$ centered at the points while $r$ grows from $0$ to $\infty$.
Indeed, Kruskal's algorithm from 1956 adds an edge whenever two components touch and merge due to the increased radius.
We generalize this construction to the chromatic setting, in which the points come in $s+1$ colors.
In particular, we introduce the \emph{$s$-lunar EMST}, which reflects the history of the connected components of the union of $s$-lunes, each the intersection of $s+1$ disks of radius $r$, one per color.
An example can be seen in Figure~\ref{fig:LunarEMST}, which shows the $1$-lunar EMST of $10$ points, $5$ of each color.
\begin{figure}[htb]
    \centering
    {%
        \setlength{\fboxsep}{0pt}%
        \setlength{\fboxrule}{.4pt}%
        \fbox{%
            \includegraphics[width=.57\linewidth]{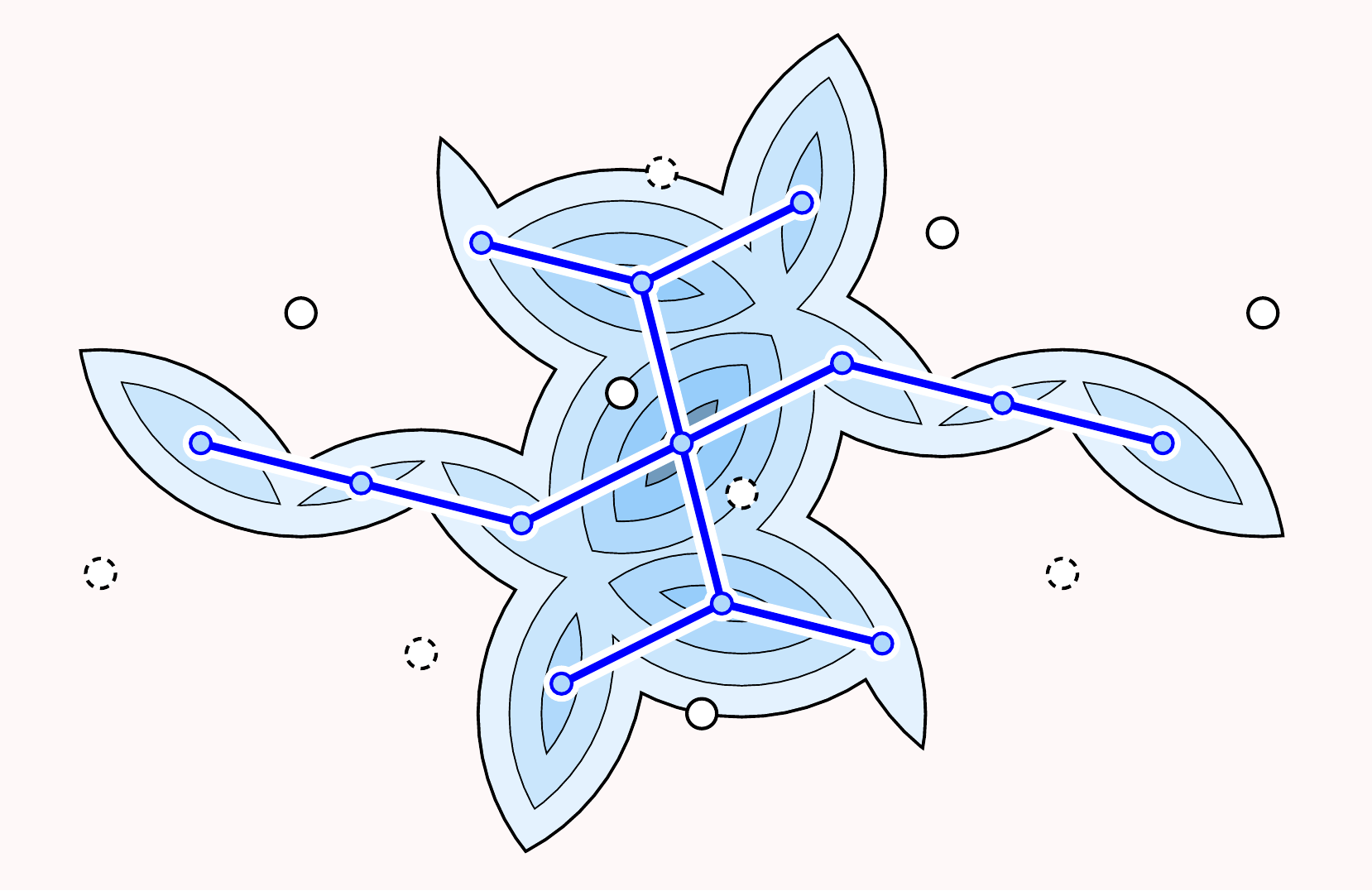}%
        }
    }%
    \caption{\footnotesize The $1$-lunar EMST, in \emph{blue}, of five \emph{dashed} and five \emph{solid} points, constructed by growing the $1$-lunes and monitoring how they meet and join into larger components.
    Only the critical nodes are shown, which are a subset of the midpoints of the given \emph{dashed} and \emph{solid} points.}
    \label{fig:LunarEMST}
\end{figure}

\smallskip
The motivation for introducing the lunar EMST is its connection to the relative diagram in the $6$-pack of chromatic persistence diagrams recently introduced in \cite{CDES26}.
Given $s+1$ colors and the map induced by the inclusion of the $s$-chromatic complex into the entire $(s+1)$-chromatic one, as defined in that paper,
we show that the degree-$s$ relative diagram in the $6$-pack is also the degree-$0$ persistence diagram of the function $f_{\max} \colon \Rspace^2 \to \Rspace$ whose sublevel sets are the unions of the $s$-lunes.
In other words, $f_{\max}$ maps every point to the maximum, over the $s+1$ colors, of the minimum distance to the points in this color.
For $s=0$, the $s$-lunes are just disks, and the $1$-norm of this degree-$0$ persistence diagram is half the cost of the (standard) EMST.
Accordingly, we will define the costs of the nodes and arcs of the $s$-lunar EMST such that half the cost of this tree is the $1$-norm of the degree-$0$ persistence diagram of $f_{\max}$.

\smallskip
If we choose the $n$ points uniformly at random in the unit square, it is natural to ask for the expected cost of the EMST.
The first serious consideration of this question dates back to just three years after the publication of Kruskal's algorithm, when Beardwood, Halton, and Hammersley \cite{BHH59} proved that the expected cost of the EMST is some constant times $\sqrt{n}$, in the limit when $n$ goes to infinity.
This work was later revisited and generalized in particular by Steele; see e.g.\ \cite{Ste97}.
While the constant is known to exist, its precise value is not known even today.
We apply the approach of Steele to the $s$-lunar EMST of points sampled uniformly at random in $[0,1]^2$ and randomly $(s+1)$-colored, and prove that there exists a constant, which depends on $s$, such that the expected cost is this constant times $\sqrt{n}$ in the limit, as $n$ goes to infinity.
Experimental results presented in \cite{DERS25} suggest that for $s+1 = 2$ colors the constant is somewhere between $0.350$ and $0.352$; compare this with the suggestion in the same paper that the constant for $s+1 = 1$ color (the standard EMST) is between $0.646$ and $0.648$.

\medskip \noindent \textbf{Outline.}
Section~\ref{sec:2} introduces the $s$-lunar EMST for an $(s+1)$-coloring of finitely many points in $\Rspace^2$ and presents three versions of this tree, all with the same cost.
This section also contains the proof that the cost is twice the $1$-norm of the degree-$s$ relative diagram in the $6$-pack.
Section~\ref{sec:3} presents background from stochastic geometry, in particular formulas for the expected number of critical simplices of the radius function on the Delaunay mosaic of a Poisson point process, and extensions to higher order and to points with weights. 
Section~\ref{sec:4} presents the main theorem of this paper and follows the general approach developed by Steele and others to prove the existence of asymptotic constants for Euclidean functionals.
Section~\ref{sec:5} concludes the paper.
In addition, Appendix~\ref{app:A} re-states and generalizes known results for stationary Poisson point processes in the Euclidean plane used in the proof of the main theorem.
Appendix~\ref{app:B} contains proofs for the combinatorial equivalence of the lunar Delaunay mosaic and the chromatic Delaunay complex, and for the filtration function being generalized discrete Morse.

\section{Lunar Generalization of the EMST}
\label{sec:2}

This section introduces the main object of study, which is the family of lunar Euclidean minimum spanning trees, or lunar EMSTs for short.
They generalize the standard Euclidean minimum spanning tree, which is a lunar EMST for only one color.
In the case of two colors, the lunar EMST may be compared with solutions to minimum matching questions, but really it minimizes the distance between pairs rather than between the points that form the pairs, as studied e.g.\ in \cite{CoTr21,Led23}.

\subsection{Growing Disks and Lunes}
\label{sec:2.1}

We generalize the geometric interpretation of Kruskal's algorithm \cite{Kru56} for a finite set $A \subseteq \Rspace^2$, which we present first.
For each point, $a \in A$, let $\Disk{r}{a}$ be the closed disk of radius $r$ centered at $a$.
We increase $r$ continuously from $0$ to $\infty$, and monitor the connected components of the union of disks.
Initially, for $r=0$, each point is a component by itself, and eventually, for sufficiently large $r$, the union is connected and thus has only one component.
The interesting events in this process are when two components merge, which happens when there are two disks---one in each component---that touch as they grow.
The distance between their centers is twice the radius at which this event happens, and the edge connecting the two centers is selected by Kruskal's algorithm when it constructs the Euclidean minimum spanning tree of $A$.

\smallskip
To generalize to the chromatic case with $s+1 \geq 1$ colors or, equivalently, to sets $A_0, A_1, \ldots, A_s \subseteq \Rspace^2$, we consider the \emph{colorful $s$-simplices}, $\aaa = \{a_0, a_1, \ldots, a_s\}$, with $a_i \in A_i$ for $0 \leq i \leq s$.
For a given radius, $r \geq 0$, this $s$-simplex corresponds to the intersection of the $s+1$ disks of radius $r$ centered at the points in $\aaa$, denoted $\Lune{r}{\aaa} = \bigcap_{i=0}^s \Disk{r}{a_i}$.
Inspired by the work on relative neighborhood graphs started by Toussaint~\cite{Tou80}, we call this intersection a \emph{lune}, or an \emph{$s$-lune} if it is important to mention the number of colors or disks involved in the intersection.
For $s=0$ it is a disk, for $s=1$ it is a convex lens bounded by two circular arcs, and in general it is a convex region with at most $s+1$ circular arcs forming its boundary.
As before, we increase $r$ continuously from $0$ to $\infty$ and monitor the connected components of the union of lunes defined by all colorful $s$-simplices formed by the $s+1$ sets.
We get an alternative description of this growth process by introducing $f_{\max} \colon \Rspace^2 \to \Rspace$ defined by mapping $x$ to the maximum, over $0 \leq i \leq s$, of the Euclidean distance to the nearest point in $A_i$.
Then $f_{\max}^{-1} [0, r]$ is the union of all $s$-lunes at radius $r$, so the growth process is equivalent to sweeping through the sublevel sets of $f_{\max}$.
To quantify the process, let $b_0 \leq b_1 \leq \ldots \leq b_m$ be the radii at which a lune becomes non-empty and starts a new component (a component is \emph{born}), and let $d_1 \leq d_2 \leq \ldots \leq d_m$ be the radii at which two components merge (a component \emph{dies}).
The \emph{cost} of the process is
\begin{align}
  \Cost{}{A_0, A_1, \ldots, A_s} &= 2 \sum\nolimits_{i=1}^m d_i - 2 \sum\nolimits_{i=1}^m b_i ,
    \label{eqn:cost}
\end{align}
where we drop $b_0$ from the second sum deliberately.
Note that $b_i \leq d_i$ for each $1 \leq i \leq m$, because we need at least $i+1$ births before the $i$-th death.
Hence $\Cost{}{A_0, A_1, \ldots, A_s} \geq 0$.

\smallskip
By de Morgan's Law, the union of all lunes defined by the $s+1$ sets is also the intersection of $s+1$ unions of disks, one per set.
If all sets are the same, $A_0 = A_1 = \ldots = A_s$, then the generalized process specializes to the original geometric interpretation of Kruskal's algorithm for $A_0$.

\smallskip
To cast the cost of the process as the solution to a combinatorial optimization problem, we construct the complete graph whose nodes are the colorful $s$-simplices defined by the $s+1$ sets.
The cost of a node is the diameter of its smallest enclosing circle, and the cost of an arc is the diameter of the smallest enclosing circle of the two nodes combined.
Running Kruskal's algorithm on this graph, we get a spanning tree that minimizes the sum of arc costs.
The cost of a node is twice the radius when the lune becomes non-empty, and that of an arc is twice the radius when the lunes of its nodes touch.
It follows that the cost of the process given in \eqref{eqn:cost} is the sum of arc costs minus the sum of node costs plus the minimum node cost of this tree.
We call this construction the \emph{combinatorial lunar EMST} of the $s+1$ sets.

\subsection{The Lunar Delaunay Mosaic}
\label{sec:2.2}

Fixing $A_0, A_1, \ldots, A_s \subseteq \Rspace^2$, we call a colorful $s$-simplex \emph{redundant} if its lune is covered by the other lunes at every radius.
The removal of a redundant $s$-simplex does not alter the course of the growth process.
Based on this observation, we produce a smaller combinatorial construction for the lunar EMST problem. Specifically, we follow \cite{BCDES26} and introduce the overlay of the Voronoi tessellations of $A_0$ to $A_s$, denoted by $\Voronoi{A_0, A_1, \dots, A_s}$.
Each cell of the tessellation is an intersection of Voronoi cells $\nu_i \in \Voronoi{A_i}$, for $0 \leq i \leq s$;
see Figure~\ref{fig:VorDel} for two examples.
We call the $2$-dimensional regions in this overlay \emph{domains}.
Each domain corresponds to a colorful $s$-simplex, $\aaa = \{a_0, a_1, \ldots, a_s\}$ with $a_i \in A_i$ for $0 \leq i \leq s$, denoted $\domain{\aaa}$, and is the set of points $x \in \Rspace^2$ such that $\Edist{x}{a_i} \leq \Edist{x}{b_i}$ for all $b_i \in A_i$ and all $0 \leq i \leq s$.
In other words, it is the intersection of the individual Voronoi domains, $\domain{\aaa} = \bigcap_{0 \leq i \leq s} \domain{a_i; A_i}$.
Its interior is also the set of points that are covered by the lune of $\aaa$ before they are covered by any other lune.
Each domain is convex, and together they decompose the union of lunes into convex sets with pairwise disjoint interiors.
Indeed
\begin{align}\label{eq:voronoi-lunes}
  \bigcup\nolimits_{\aaa} \Lune{r}{\aaa} &= \bigcup\nolimits_{\aaa} \left[ \Lune{r}{\aaa} \cap \domain{\aaa} \right]
\end{align}
for every $r \geq 0$, in which the unions are over all colorful $s$-simplices of $A_0, A_1, \ldots, A_s$.
This implies that a colorful $s$-simplex without corresponding domain is redundant.
Conversely, a colorful $s$-simplex with non-empty domain is not redundant since for every interior point in the domain there is a radius for which the corresponding lune is the only lune that contains the point.
Note, however, that the point from which the lune starts growing might be outside the domain.
\begin{figure}[hbt]
    \centering
    {%
        \setlength{\fboxsep}{0pt}%
        \setlength{\fboxrule}{.4pt}%
        \fbox{%
            \includegraphics[width=.49\linewidth]{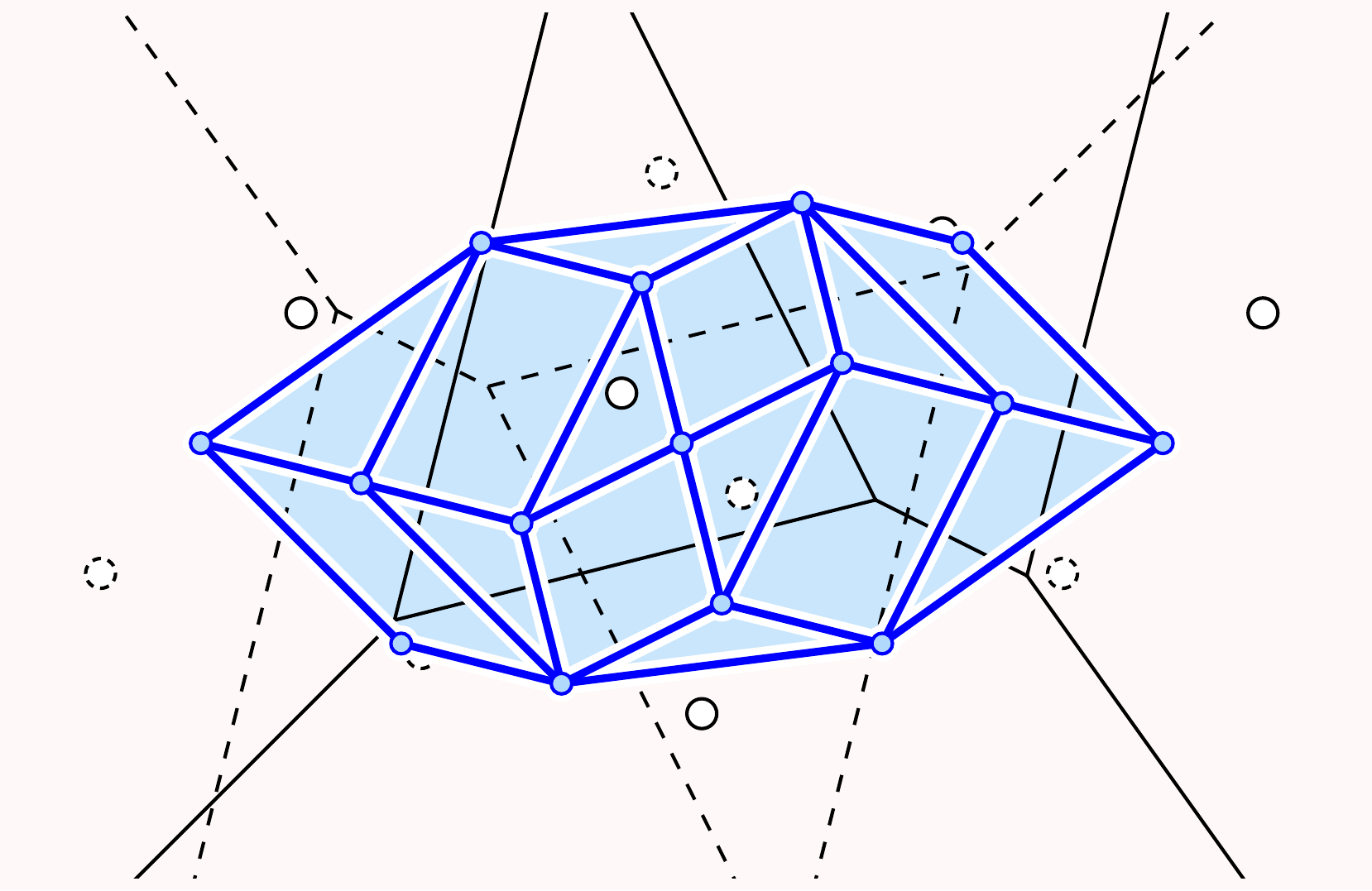}%
        }
    }\hspace{.005\linewidth}%
    {%
        \setlength{\fboxsep}{0pt}%
        \setlength{\fboxrule}{.4pt}%
        \fbox{%
            \includegraphics[width=.49\linewidth]{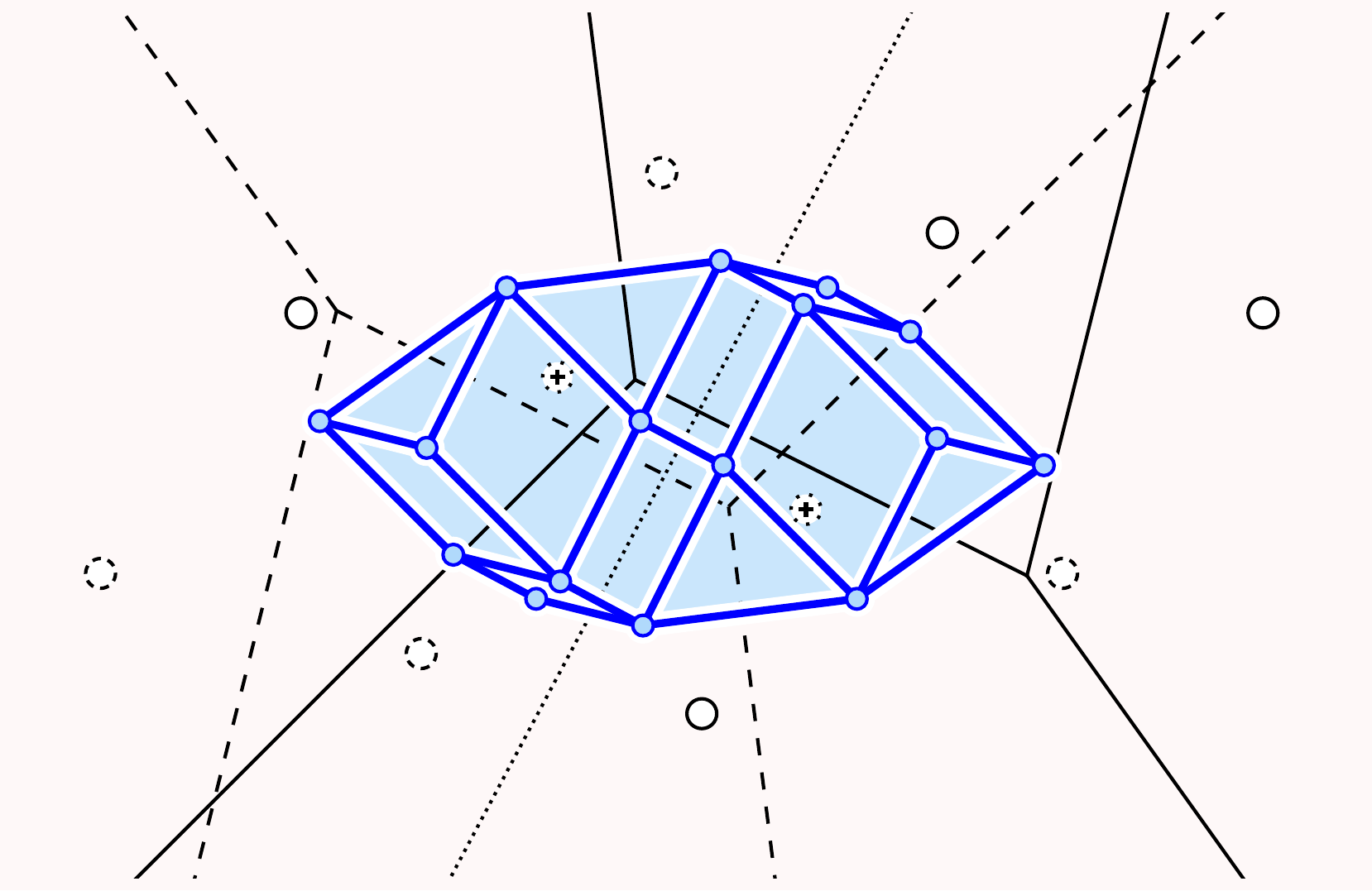}%
        }
    }%
    \caption{\footnotesize Two lunar Delaunay mosaics superimposed on the corresponding overlays of Voronoi tessellations.
    On the \emph{left}, we have two tessellations, and the (\emph{blue}) nodes of the mosaic are midpoints of \emph{solid} and \emph{dashed} points; compare with Figure~\ref{fig:LunarEMST}.
    On the \emph{right}, we have three tessellations, and the (\emph{blue}) nodes of the mosaic are centroids of triangles spanned by a \emph{solid}, a \emph{dashed}, and a point marked by a \emph{plus}.
    The \emph{blue} arcs of the mosaics connect the \emph{blue} nodes of neighboring domains in the overlays.}
    \label{fig:VorDel}
\end{figure}

\smallskip
By focusing exclusively on the $s$-simplices that are not redundant, we reduce the size of the graph we study considerably.
To count the domains, we suppose the cardinality of each $A_i$ is $n$, and that the $(s+1)n$ points are in general position.\footnote{
    See Definition~\ref{dfn:strong_chromatic_genericity} in the appendix for the definition of generic position.
}
Every one of the $s+1$ Voronoi tessellations has fewer than $3n$ edges, and each edge crosses each edge in the other $s$ tessellations in at most one point.
It follows that the number of crossings is less than some constant times $s^2 n^2$.
Since each tessellation has fewer than $2n$ vertices, the total number of vertices in the overlay is at most some constant times $s^2 n^2$, each belonging to either three or four edges.
Hence, also the number of edges and domains of the overlay is at most some constant times $s^2 n^2$;
see \cite{BCDES26} for more general bounds that subsume the case considered here.

\smallskip
This motivates us to introduce the \emph{lunar Delaunay mosaic}, denoted $\Delaunay{A_0, A_1, \ldots, A_s}$, as the dual of the overlay of the $s+1$ Voronoi tessellations.
Specifically, we have a node for every domain, an arc for every pair of domains that share a common side, a quadrangle for every crossing of two Voronoi edges, and a triangle for every Voronoi vertex---which is an exhaustive list of possibilities, assuming the points are in a general position.
Using an idea in \cite{Aur90}, we map each node to the average of its $s+1$ points to geometrically realize the lunar Delaunay mosaic.
For $s+1=2$, each lune starts growing from this point, but for $s+1 > 2$ this is not necessarily the case.
As illustrated in Figure~\ref{fig:VorDel}, this implies that the arcs of the Delaunay mosaic are orthogonal to the corresponding edges in the overlay, and the quadrangles are in fact parallelograms.
Recall the function $f_{\max} \colon \Rspace^2 \to \Rspace$, which maps every $x \in \Rspace^2$ to the minimum radius for which $x$ belongs to the union of lunes.
We use it to define the \emph{lunar radius function}, $\Lambda \colon \Delaunay{A_0, A_1, \ldots, A_s} \to \Rspace$, which maps each node, arc, triangle, and parallelogram to the minimum value of $f_{\max}$ of any point in its corresponding domain, edge, vertex, and crossing in the overlay of Voronoi tessellations.
Using the Nerve Theorem for convex sets on the clipped lunes (the terms within the union on the right-hand side of \eqref{eq:voronoi-lunes}), it is not difficult to prove that $\Lambda^{-1} [0,r] = \Delaunay[r]{A_0, A_1, \dots, A_s}$ has the same homotopy type as $f_{\max}^{-1} [0,r]$.
We can therefore use the lunar radius function as a proxy in the study of the growth process of lunes:
\begin{lemma}
  \label{lem:lunar_alpha_vs_lunar_space}
  There is a homotopy equivalence,
  $\Delaunay[r]{A_0, A_1, \dots, A_s} \simeq \bigcup\nolimits_{\aaa\in A_0 \times \dots \times A_s} \Lune{r}{\aaa}$, that is natural with respect to the inclusions for the growing parameter $r$.
\end{lemma}

\smallskip
At this time, we are primarily interested in the $1$-skeleton of the lunar Delaunay mosaic, which is a graph whose nodes and arcs have radii assigned to them by $\Lambda$.
Note that this is a subgraph of the graph we used to construct the combinatorial lunar EMST, but the lunar radius of a node or arc is not necessarily half the cost assigned to it, as in the earlier construction.
Nevertheless, by Lemma~\ref{lem:lunar_alpha_vs_lunar_space}, we can run Kruskal's algorithm on the weighted $1$-skeleton and get a spanning tree whose cost is again $\Cost{}{A_0, A_1, \ldots, A_s}$ as given in \eqref{eqn:cost}.
We call this the \emph{geometric lunar EMST}, which is considerably smaller than its combinatorial counterpart due to culling nodes and arcs with the overlay of Voronoi tessellations.

\subsection{Critical and Non-critical Cells}
\label{sec:2.3}

Even non-redundant nodes are not necessarily relevant to solve the optimization problem.
We therefore restrict ourselves further to the \emph{critical cells}, which are the nodes, arcs, and triangles in the lunar Delaunay mosaic associated with events that change the homotopy type of the union of lunes.
By definition of criticality, the lunar radii of the critical nodes and arcs agree with half the costs assigned to them in Section~\ref{sec:2.1}.
Triangles in the lunar Delaunay mosaic may or may not be critical, but parallelograms are necessarily non-critical, as we will see shortly.
We will need a geometric characterization of critical cells.
For this purpose, we let $\xxx$ be a colorful simplex in the cross-product of finite sets $A_0, A_1, \ldots, A_s$ in $\Rspace^2$, by which we mean that $\xxx$ contains at least one point of each color, so its dimension is $s$ or larger. 
The \emph{smallest enclosing stack} of $\xxx$ is a collection of $s+1$ concentric circles---one for each color---such that
\smallskip \begin{itemize}
  \item the largest circle in the stack is the smallest enclosing circle of the points in $\xxx$;
  \item the $i$-th circle in the stack encloses all points of $\xxx \cap A_i$, for $0 \leq i \leq s$.
\end{itemize} \smallskip
The \emph{size} of the stack is the radius of its largest circle, and the stack is \emph{empty} if all points of $A_i\setminus \xxx$ lie outside and all points of $A_i\cap\xxx$ lie on its $i$-th circle, for $0 \leq i \leq s$.
Observe that the center of the smallest enclosing stack is the first point of intersection of the growing lunes defined by the points in $\xxx$.
\begin{lemma}
  \label{lem:critical_and_non-critical_cells}
  Let $A_0, A_1, \ldots, A_s$ be finite sets of points in general position in $\Rspace^2$, and $\aaa, \bbb, \ccc$ colorful $s$-simplices in the cross-product of the $s+1$ sets.
  Then 
  \smallskip \begin{enumerate}[(1)]
    \item $\aaa$ is a critical node iff its smallest enclosing stack is empty;
    \item $\aaa, \bbb$ are the nodes of a critical arc iff $\card{(\aaa \cap \bbb)} = s$, and the smallest enclosing stack of $\aaa \cup \bbb$ is empty and strictly larger than the smallest enclosing stacks of $\aaa$ and $\bbb$;
    \item $\aaa, \bbb, \ccc$ are the nodes of a critical triangle iff $\card{(\aaa \cap \bbb \cap \ccc)} = s$, and the smallest enclosing stack of $\aaa \cup \bbb \cup \ccc$ is empty and strictly larger than the smallest enclosing stacks of $\aaa \cup \bbb$, $\bbb \cup \ccc$, $\ccc \cup \aaa$.
  \end{enumerate}
\end{lemma}
\begin{proof}
  (1) Consider the smallest radius, $r$, for which the closed disks centered at the points of $\aaa$ have a non-empty common intersection, and let $x = x(\aaa)$ be the point at which the disks intersect.
  Then $\aaa$ is a critical node iff $x$ lies in the interior of $\domain{\aaa}$---indeed, otherwise another lune is covering $x$ at radius $r$, and $\aaa$ does not create a new connected component.
  Writing $a_i = \aaa \cap A_i$, this is the case iff $x$ lies in the interior of $\domain{a_i; A_i}$, for each $0 \leq i \leq s$, which happens iff the smallest enclosing stack of $\aaa$ is empty.

  \smallskip
  (2) Similarly, let $y = y(\aaa \cup \bbb)$ be the first point at which the closed disks centered at the points in $\aaa \cup \bbb$ intersect.
  Then the arc connecting the nodes $\aaa$ and $\bbb$ is critical iff $x(\aaa)$ and $x(\bbb)$ are defined for strictly smaller radii than $y$, and $y$ lies in the interior of the side shared by $\domain{\aaa}$ and $\domain{\bbb}$.
  Letting $i$ be the color such that this side is part of the bisector of two points in $A_i$, we have $\card{(\aaa \cap \bbb)} = s$ as they share the points of the other $s$ colors.
  Observe that $a_i = \aaa \cap A_i$ and $b_i = \bbb \cap A_i$ lie on the largest circle in the stack, else the stacks of $\aaa$ and $\bbb$ would not be strictly smaller than that of $\aaa \cup \bbb$.
  Hence, $y$ lies in the interior of the shared side iff the smallest enclosing stack of $\aaa \cup \bbb$ is empty.
  
  \smallskip
  (3) The argument for arcs generalizes to triangles, which we therefore omit.
\end{proof}

We note that a parallelogram in the lunar Delaunay mosaic cannot be critical for points in general position.
To see this, let $x$ be the corresponding crossing of two Voronoi edges, say the edges that separate the regions of $a_0, b_0 \in A_0$ and the regions of $a_1, b_1 \in A_1$.
If $s+1=2$ and the two sets contain no additional points, then the intersection of the four lunes defined by the two plus two points is also the intersection of two unions of two disks each.
Neither of these two unions has a hole, so nor does their intersection.
Indeed, if it had a hole, then there would be a loop contained in the intersection that surrounds the hole, but this loop is filled in both unions of two disks, so it is also filled in their intersection.
In the more general case, when there are additional points---in $A_0$ or $A_1$ or in additional sets---we assume that no two of them have the same distance to $x$, and all these distances differ from $\Edist{x}{a_0} = \Edist{x}{b_0}$ and $\Edist{x}{a_1} = \Edist{x}{b_1}$.
So setting $r = f(x)$, the points at distance $r$ from $x$ are either two or all four of $a_0, b_0, a_1, b_1$, or a unique other point.
In none of these cases, we see a hole filled up around $x$ as we go from a radius smaller to a radius larger than $r$.

\begin{figure}
    \centering
    {%
        \setlength{\fboxsep}{0pt}%
        \setlength{\fboxrule}{.4pt}%
        \fbox{%
            \includegraphics[width=.48\linewidth]{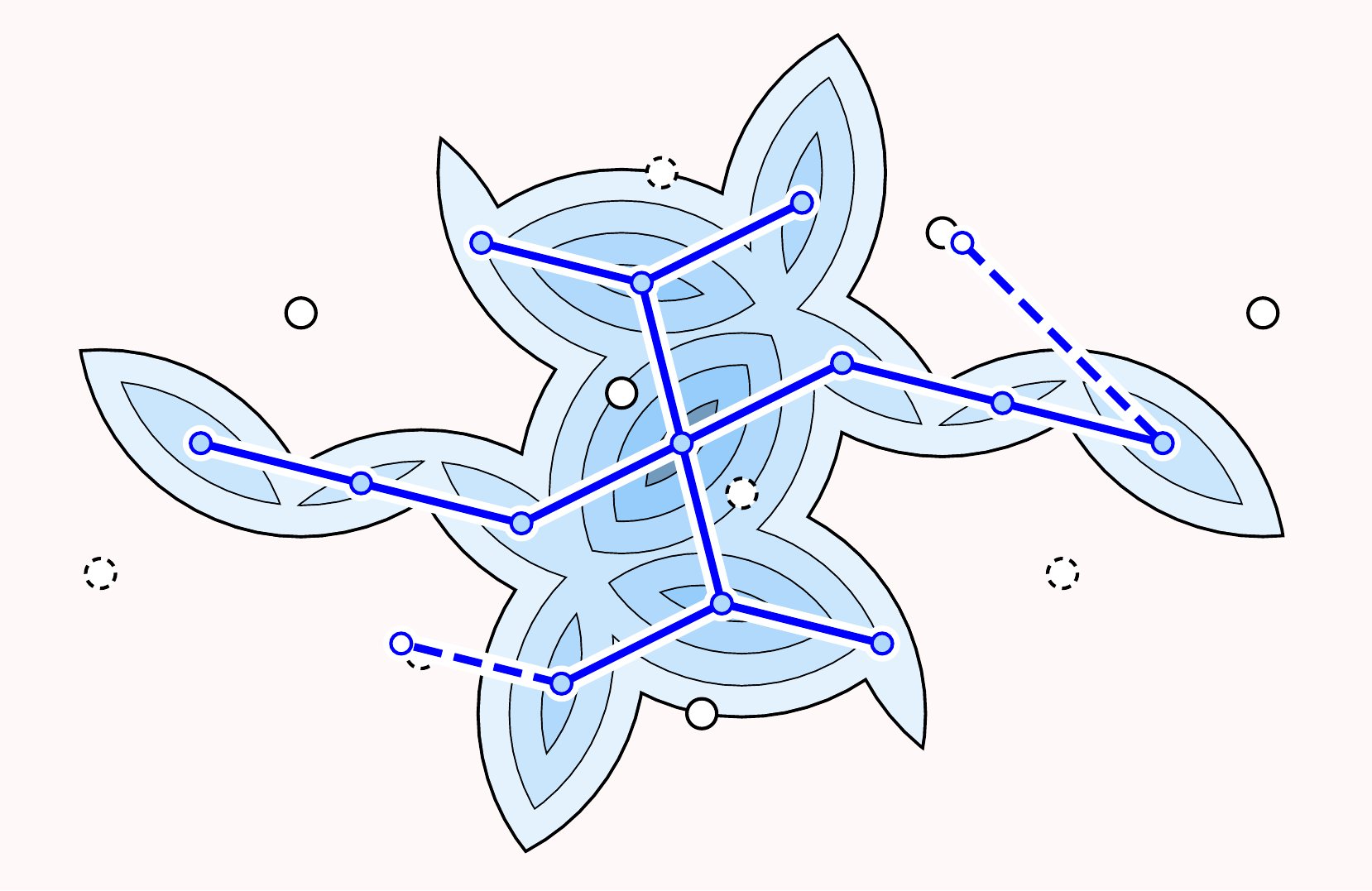}%
        }
    }\hspace{.01\linewidth}%
    {%
        \setlength{\fboxsep}{0pt}%
        \setlength{\fboxrule}{.4pt}%
        \fbox{%
            \includegraphics[width=.48\linewidth]{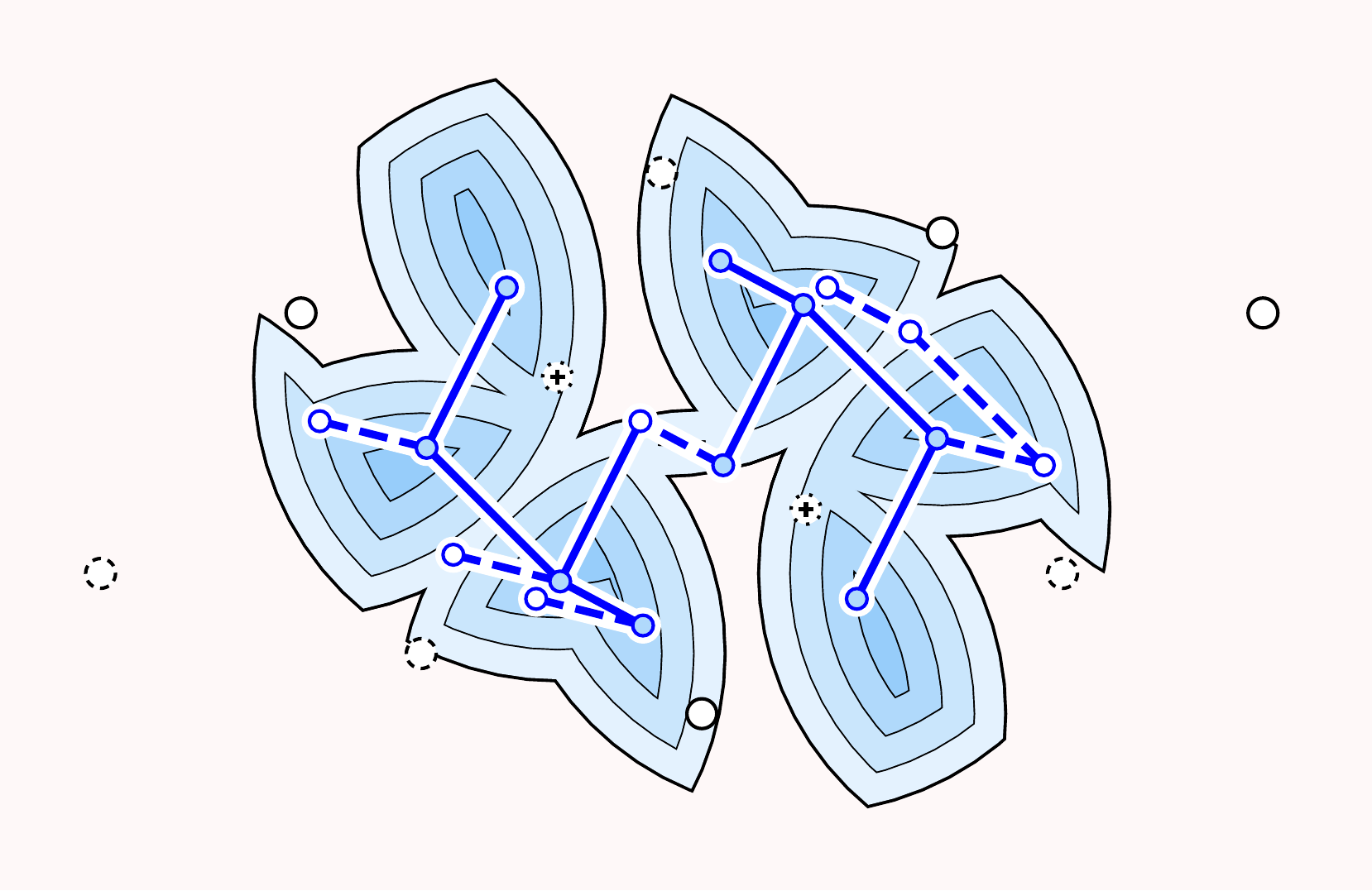}%
        }
    }%
    \caption{\footnotesize On the \emph{left} the $1$-lunes and on the \emph{right} the $2$-lunes of the points in Figure~\ref{fig:VorDel}, each together with corresponding lunar EMST.
    The \emph{solid blue} arcs and \emph{blue} nodes are critical, and the \emph{dashed blue} arcs and \emph{white} nodes are not. 
    While the points are symmetric, the arbitrary breaking of ties in the choice of non-critical arcs is responsible for the lack of symmetry in the constructed trees.}
    \label{fig:LunarEMST-with-noncritical}
\end{figure}

\smallskip
We prove upper bounds for the number of critical nodes, arcs, and triangles that are significantly lower than the bounds established for the nodes, arcs, and domains of the lunar Delaunay mosaic.

\begin{lemma}
  \label{lem:counting_critical_cells}
  Let $A_0, A_1, \ldots, A_s$ be sets of $n$ points in $\Rspace^2$ each.
  Then the number of critical nodes, arcs, and triangles of the lunar radius function on $\Delaunay{A_0, A_1, \ldots, A_s}$ is at most $O(n)$.
\end{lemma}
\begin{proof}
  For $s+1 = 1$, the Delaunay mosaic has at most $n$ vertices, $3n$ edges, and $2n$ triangles, which includes all critical vertices, edges, and triangles. 
  Henceforth assume $s+1 \geq 2$.

  \smallskip
  By Lemma~\ref{lem:critical_and_non-critical_cells}, a node is critical only if its smallest enclosing stack is empty.
  Assuming general position of the points, the largest circle of this stack---which is also the smallest enclosing circle of the node---passes through either two or three of the $s+1$ points. Suppose two points.
  If they belong to $A_i$ and $A_j$, then they span a critical edge in the (standard) Delaunay mosaic of $A_i \cup A_j$.
  There are $\binom{s+1}{2}$ choices of two sets and at most $6n$ critical edges in each Delaunay mosaic of $2n$ points, which implies an upper bound of $3 (s+1) s n$ such cases.
  Next suppose three points on the circle.
  If they belong to $A_i, A_j, A_k$, then they span a critical triangle in the (standard) Delaunay mosaic of $A_i \cup A_j \cup A_k$.
  We have $\binom{s+1}{3}$ choices of three sets and at most $6n$ critical triangles in each Delaunay mosaic of $3n$ points, which implies an upper bound of $(s+1) s (s-1) n$ such cases.
  There is no double-counting since a smallest enclosing stack uniquely defines the corresponding node.
  Adding the two bounds, we get $(s+2)(s+1)s n = O(n)$ as an upper bound on the number of critical nodes.

  \smallskip
  An arc of the lunar Delaunay mosaic is critical only if the smallest enclosing stack of the corresponding $s+2$ points is empty.
  Assume $A_i$ contributes two to the $s+2$ points while every other set contributes only one point.
  Since the smallest enclosing stack of the arc is strictly larger than those of its two nodes, both points contributed by $A_i$ lie on the largest circle in the stack.
  If they are the only points on the circle, then they span a critical edge in the (standard) Delaunay mosaic of $A_i$, and if there is a third point on the circle, say contributed by $A_j$, then the three points span a critical triangle in the (standard) Delaunay mosaic of $A_i \cup A_j$.
  So the number of cases is at most $(s+1) 3n + \binom{s+1}{2} 4n = O(n)$.

  \smallskip
  Finally, a triangle of the lunar Delaunay mosaic is critical only if the smallest enclosing stack of the corresponding $s+3$ points is empty.
  Assume $A_i$ contributes three to the $s+3$ points, while every other set contributes only one.
  Since the smallest enclosing stack of the triangle is strictly larger than those of its arcs, all three points contributed by $A_i$ lie on the circle, so they span a critical triangle in the (standard) Delaunay mosaic of $A_i$.
  The number of such cases is therefore at most $2 (s+1) n = O(n)$.
\end{proof}

Since only critical nodes and arcs of the lunar Delaunay mosaic contribute to the cost of the lunar EMST, it should be possible to dispense with all non-critical nodes and arcs.
One way to do this starts with the geometric lunar EMST and removes every non-critical arc by canceling it with an incident non-critical node.
This operation contracts the arc to the other incident node, which may or may not be critical, and implicitly merges the two nodes into one.
We call the result the \emph{topological lunar EMST} and note that its cost is again $\Cost{}{A_0, A_1, \ldots, A_s}$ as given in \eqref{eqn:cost}.
By Lemma~\ref{lem:counting_critical_cells}, it is considerably smaller than its combinatorial and geometric counterparts.
  
\subsection{Relation to Chromatic Persistence}
\label{sec:2.4}

As mentioned in the introduction, the original motivation for the lunar generalizations of the Euclidean minimum spanning tree is the theory of chromatic persistence introduced in \cite{CDES26}.
In particular, the lunar EMST of an $(s+1)$-coloring of finitely many points in Euclidean space corresponds to the degree-$s$ relative persistence module of the inclusion of all combinations of $s$ colors in $s+1$ colors, as we will now explain.

\smallskip
Let $A$ be a finite set of points in general position in $\Rspace^2$.
For an integer $s \geq 0$, we write $[s] = \{0, 1, \ldots, s\}$, and letting $\chi \colon A \to [s]$ be an $(s+1)$-coloring, we write $A_i = \chi^{-1} (i)$, for $0 \leq i \leq s$.
We recall from \cite{CDES26} that the \emph{chromatic Delaunay mosaic}, denoted $\Delaunay{\chi}$, is the simplicial complex that contains a simplex $\sigma \subseteq A$ whenever the Voronoi domains of its vertices have a non-empty intersection; that is:
\begin{align}
  \bigcap\nolimits_{a \in \sigma} \domain{a; A_{\chi(a)}} &\neq \emptyset .
  \label{eqn:domains}
\end{align}
It is filtered by the \emph{chromatic radius function}, $\alpha \colon \Delaunay{\chi} \to \Rspace$, which maps $\sigma$ to the minimum radius, $r$, for which at least one point of the intersection of Voronoi domains \eqref{eqn:domains} is contained in the union of disks of radius $r$ centered at the vertices of $\sigma$.
For $a \in A$ with color $i = \chi(a)$, we write $D_r (a)$ for the closed disk of radius $r$ with center $a$, and call $D_r(a) \cap \domain{a; A_i}$ the corresponding \emph{Voronoi disk}.
The filtration induced by $\alpha$ is the sequence of complexes $\Alpha{r}{\chi} = \alpha^{-1} [0,r]$, and we note that $\Alpha{r}{\chi}$ is indeed isomorphic to the nerve of the Voronoi disks of radius $r$ centered at the points $a \in A$.
Call a simplex $\sigma \subseteq A$ \emph{colorful} if it contains at least one point of each color in $[s]$, and write
\begin{align}
  \Delaunaycol{\chi} &= \left\{ \sigma \in \Delaunay{\chi} \mid \chi (\sigma) = [s] \right\}
\end{align}
for the collection of colorful simplices, which we note is not necessarily a complex; see Figure~\ref{fig:LunarDelaunay-LiftedColorfulDelaunay}.
Nevertheless, we can define a sequence of subcollections by restricting $\alpha$ to $\Delaunaycol{\chi}$ and take sublevel sets of this restriction.
We also let $\Delaunaycol{\chi}$ inherit the boundary operator, $\partial$, from $\Delaunay{\chi}$ by ignoring all simplices that are not colorful.
For each radius, $r$, we thus have a combinatorial collection of cells, partitioned by dimension, with boundary operator and therefore a well-defined chain complex.\footnote{Such a combinatorial structure is often referred to as a \emph{Lefschetz complex}, due to \cite[Ch.\ III, Sec.\ 1, Def.\ 1.1]{Lef42}.}
Importantly, the homology of this chain complex reflects the relative homology of a simplicial pair, as we discuss shortly.
Assuming general position of the points in $A$, we claim that $\Delaunaycol{\chi}$ combinatorially agrees with the lunar Delaunay mosaic, $\Delaunay{A_0, A_1, \dots, A_s}$.
\begin{proposition}
  \label{prop:chromatic_vs_lunar}
  Let $A$ be a finite set of points in general position in $\Rspace^2$, and $\chi \colon A \to [s]$ an $(s+1)$-coloring.
  Then there is a bijection $\varphi \colon \Delaunaycol{\chi} \to \Delaunay{A_0, A_1, \dots, A_s}$ between the colorful simplices of the chromatic alpha complex and the cells in the lunar Delaunay mosaic, such that
  \begin{itemize}
    \item $\dim \sigma = \dim \varphi(\sigma) + s$;
    \item $\alpha(\sigma) = \Lambda(\varphi(\sigma))$;
    \item $\varphi(\partial\sigma) = \partial\varphi(\sigma)$,
  \end{itemize}
  for each $\sigma \in \Delaunaycol{\chi}$, in which $\alpha$ and $\Lambda$ are the radius functions of the two mosaics.
\end{proposition}

We refer to Appendix~\ref{app:B} for a proof of this proposition.
There, we also discuss the construction for point sets in arbitrary dimensions and, furthermore, argue that under a sufficiently strong notion of general position, $\Lambda$ is a generalized discrete Morse function; see \cite{For98,Fre09} for background on such functions.
To give geometric meaning to $\Delaunaycol{\chi}$, we recall that the \emph{$s$-chromatic subcomplex} of $\Delaunay{\chi}$ consists of all simplices whose vertices lack at least one color:
\begin{align}
  \Delaunay{\chi; s} &= \left\{ \sigma \in \Delaunay{\chi} \mid \card{\chi(\sigma)} \leq s \right\} ;
\end{align}
see \cite[Section~3.4]{CDES26}.
Its radius function is the restriction of $\alpha$ to its simplices.
The relative persistent homology of the filtered pair $\left( \Delaunay{\chi}, \Delaunay{\chi; s} \right)$ is defined via the relative chain complex, which is isomorphic to the restriction of the chain complex of $\Delaunay{\chi}$ to $\Delaunaycol{\chi}$.
Together with Lemma~\ref{lem:lunar_alpha_vs_lunar_space} and Proposition~\ref{prop:chromatic_vs_lunar}, this reduction explains the correspondence between the degree-$s$ relative persistence module of the filtered pair $\left(  \Delaunay{\chi}, \Delaunay{\chi; s} \right)$ and the lunar EMST.
\begin{corollary}
  \label{cor:cost_and_norm}
  Let $A$ be a finite set of points in general position in $\Rspace^2$, $\chi \colon A \to [s]$ an $(s+1)$-coloring, and $\Lspace_r$ the union of lunes of radius $r$ defined by the $s$-simplices in $A_0 \times A_1 \times \ldots \times A_s$.
  The degree-$p$ persistence module of the filtration of the spaces $\Lspace_r$ is isomorphic to the degree-$(p+s)$ relative persistence module of the filtration of pairs $( \Alpha{r}{\chi}, \Alpha{r}{\chi;s} )$, in which $\Alpha{r}{\chi;s} = \Alpha{r}{\chi} \cap \Delaunay{\chi; s}$.
  In particular, $\Cost{}{A_0, A_1, \ldots, A_s}$ is twice the $1$-norm of the degree-$s$ relative persistence diagram of the pairs.
\end{corollary}

\begin{figure}[hbbt]
    \centering
    {%
        \setlength{\fboxsep}{0pt}%
        \setlength{\fboxrule}{0pt}%
        \fbox{%
            
            \raisebox{0cm}{\includegraphics[width=.65\linewidth]{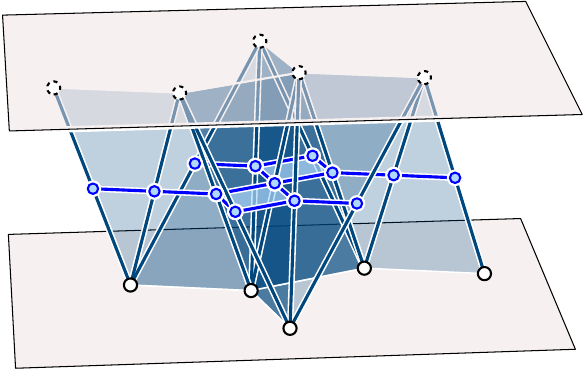}}%
        }
    }%
    \vspace{-0.0in}
    \caption{\footnotesize The connection between the lunar Delaunay mosaic and the colorful simplices in the chromatic Delaunay mosaic.
    The points come in two colors and are placed on two parallel planes connected by the colorful edges, triangles, and tetrahedra.
    The \emph{bright blue} middle section is isomorphic to the subcomplex of the lunar Delaunay mosaic for the given radius.
    It contains a vertex for each colorful edge, an edge for each colorful triangle, and two quadrangles for the two colorful tetrahedra.}
    \label{fig:LunarDelaunay-LiftedColorfulDelaunay}
\end{figure}

\section{Poisson--Delaunay Mosaics}
\label{sec:3}

To analyze the lunar EMSTs and to compare different kinds of such trees, we will need bounds on the expected number of cells in a Delaunay mosaic as well as their expected sums of radii and squared radii.
A focus will be the interaction between the trees and the boundary of the region of interest, which is primarily the unit square in the plane, $[0,1]^2$.
For standard Delaunay mosaics, we derive these bounds from the generalization of Theorem~1 in \cite{ENR17} proved in Appendix~\ref{app:A.1}, and its extensions in Appendix~\ref{app:A.2}, and for weighted Delaunay mosaics, we derive them from Theorem~1 in \cite{EdNi19} restated in Appendix~\ref{app:A.4}, and its extensions in Appendix~\ref{app:A.5}.

\subsection{Ordinary Delaunay Mosaics}
\label{sec:3.1}

An edge or triangle in the (ordinary) Delaunay mosaic of a locally finite set in $\Rspace^2$ is \emph{critical} if the smallest enclosing circle is also the circumcircle of the edge or triangle.
Given a stationary Poisson point process on $\Rspace^2$, we can use integral geometry to determine the number of such edges and triangles whose circumcircles have the center inside $[0,1]^2$, and also the sum of their radii or squared radii.
The same analysis can be used for a Poisson point process in $[0,1]^2$, but near the boundary of the unit square, the circumscribed circle may lie partially outside $[0,1]^2$, so the probability that it does not enclose other points is higher than if the circle is entirely inside $[0,1]^2$.
However, since the circle is determined by points inside the square, at least half the area of the disk bounded by the circle is inside the square.
By requiring that the intersection of the disk with $[0,1]^2$ be empty of sampled points, we can again use integral geometry to get a weaker version of Theorem~1 in \cite{ENR17} and thus worse bounds.
Nevertheless, our bounds are of the same order of magnitude.
Similarly we can use integral geometry to get bounds of the same order of magnitude if we allow one point inside the circumcircle.
\begin{lemma}
  \label{lem:radii_and_squared_radii}
  Let $A$ be a stationary Poisson point process with intensity $n$ in $\Rspace^2$ or in $[0,1]^2$.
  Then the expected sum of radii of the critical edges and triangles with circumcircles that are centered in $[0,1]^2$ and enclose at most one point of $A$ is at most a constant times $\sqrt{n}$, and the expected sum of their squared radii is at most a constant.
\end{lemma}
We refer to Appendix~\ref{app:A} for more general theorems that imply this lemma, and to Appendix~\ref{app:A.3} for a more detailed statement of the lemma and its various cases.
We use the bounds for the squared radii to show that the expected sum of radii is bounded by a constant if we limit ourselves to the circles that cross the boundary of $[0,1]^2$.
\begin{lemma}
  \label{lem:crossing_the_boundary}
  Let $A$ be a stationary Poisson point process with intensity $n$ in $\Rspace^2$ or $[0,1]^2$.
  Then the expected sum of radii of the critical edges and triangles whose smallest enclosing circles enclose at most one point of $A$, have the center inside $[0,1]^2$, and cross the boundary of $[0,1]^2$ is at most $O(1)$. 
\end{lemma}
We again refer to Appendix~\ref{app:A} for more general theorems that imply this lemma, and to Lemma~\ref{lem:crossing_the_boundary_details} and Table ~\ref{tbl:bounds} in Appendix~\ref{app:A.3} for the constants in the bounds for each of the different cases.

\subsection{Weighted Delaunay Mosaics}
\label{sec:3.2}

To understand how a lunar EMST is formed near the boundary of $[0,1]^2$, we study the intersection of a line with a Voronoi tessellation in $\Rspace^2$.
While it might seem far-fetched, we find it useful to think of this intersection as a $1$-dimensional \emph{weighted Voronoi tessellation}.
To explain, let $A \subseteq \Rspace^2$ be locally finite, and grow a disk centered at each point, which sweeps out its \emph{domain} in the Voronoi tessellation of $A$.
Given any line, denoted $\Rspace^1 \hookrightarrow \Rspace^2$, these domains decompose $\Rspace^1$ into \emph{segments}, and the growing disks sweep out these segments; see Figure~\ref{fig:weighted}.

\smallskip
We can describe everything that happens along the line intrinsically: the \emph{weighted points} that generate the segments are the orthogonal projections of points in $A$, weighted by their squared distance to the line, and each segment is swept out by a growing $1$-ball centered at the generating point.
Letting $A'$ be the set of weighted points in $\Rspace^1$, we write $\Delaunay{A'}$ for the \emph{weighted Delaunay mosaic}, which is dual to the decomposition of $\Rspace^1$ into segments.
In analogy to the unweighted case, we introduce the \emph{weighted radius function}, $g \colon \Delaunay{A'} \to \Rspace$, which maps every vertex, $a' \in A'$, to the minimum radius at which the disk centered at the corresponding point $a \in A$ touches a point of the corresponding segment.
Similarly, $g$ maps the edge with endpoints $a'$ and $b'$ to the minimum radius at which the disks centered at $a$ and $b$ reach the shared endpoint of the two corresponding segments.
Since the points in $\Rspace^2$ are some non-negative distance from $\Rspace^1$, the radius $g$ assigns to any vertex or edge in $\Delaunay{A'}$ is necessarily non-negative.
\begin{figure}[hbt]
  \centering
  \vspace{0.0in}
  \resizebox{!}{1.6in}{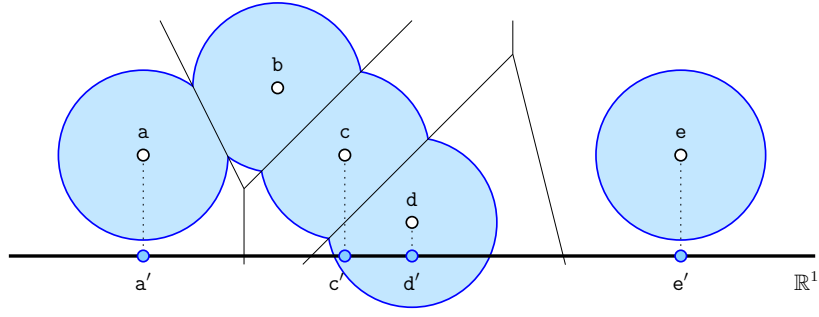}
  \caption{\footnotesize The horizontal line intersects the Voronoi tessellation of five points in a ($1$-dimensional) weighted Voronoi tessellation of the projections of these points.
  The latter has only four segments because the domain of $\tt b$ does not intersect the line.
  The domain of $\tt c$ intersects the line in a segment, but the projection $\tt c'$ of $\tt c$ onto the line does not lie inside this segment.
  By contrast, the segments of $\tt a$, $\tt d$, and $\tt e$ contain the corresponding vertices of the weighted mosaic, which are $\tt a'$, $\tt d'$, and $\tt e'$.}
  \label{fig:weighted}
\end{figure}

\smallskip
We distinguish between a \emph{critical vertex}, which is a weighted point inside its corresponding segment, a \emph{critical edge}, which contains the shared endpoint of the corresponding two segments, 
and a \emph{non-critical vertex-edge pair}, where the weighted point shares its weighted radius with the incident edge.
To facilitate counting, we call the location of a vertex and the shared endpoint of the two segments the \emph{center} of the vertex and edge in $\Delaunay{A'}$, respectively. 
\begin{lemma}
  \label{lem:weighted_radii}
  Let $A$ be a stationary Poisson point process with positive intensity in $\Rspace^2$, $\Rspace^1 \hookrightarrow \Rspace^2$ a fixed line, $A'$ the weighted points obtained by projecting $A$ orthogonally to $\Rspace^1$, $g \colon \Delaunay{A'} \to \Rspace$ the weighted radius function, and $\Omega = [0,1] \subseteq \Rspace^1$.
  Then the expected sum of weighted radii of the critical vertices, non-critical vertex-edge pairs, and critical edges with centers in $\Omega$ is at most a constant.
\end{lemma}
We refer to Appendix~\ref{app:A} for a more general theorem that implies this lemma, and to Table~\ref{tbl:weighted_bounds} in Appendix~\ref{app:A.5} for the constants.
We will repeatedly make use of an extension of this lemma to the $(s+1)$-colored case, which we now formalize.
Let $\chi \colon A \to [s]$, write $A_i = \chi^{-1} (i)$, recall that $\Voronoi{A_0, A_1, \ldots, A_s}$ is the overlay of $s+1$ Voronoi tessellations---one for each color---and observe that its intersection with $\Rspace^1$ is the overlay of the corresponding $s+1$ weighted Voronoi tessellations, denoted $\Voronoi{A_0', A_1', \ldots, A_s'}$.
This overlay decomposes $\Rspace^1$ into intervals we call \emph{pieces}, each the common intersection of $s+1$ segments, one from each weighted Voronoi tessellation.
The relevant distance function is $f_{\max} \colon \Rspace^1 \to \Rspace$, which maps every $x \in \Rspace^1$ to the maximum---over the $s+1$ colors---of the Euclidean distance to the nearest point of this color.
Hence, $f_{\max} (x)$ is the radius at which the union of lunes reaches $x$.

\smallskip
We write $g \colon \Delaunay{A_0', A_1', \ldots, A_s'} \to \Rspace$ for the radius function on the dual of the overlay of weighted Voronoi tessellations, called \emph{weighted lunar Delaunay mosaic}, that maps each vertex and edge to the minimum radius at which the corresponding piece and shared endpoint of two pieces is touched by the union of lunes.
A \emph{critical vertex} of $g$ corresponds to a piece that is touched first at an interior point, and a \emph{critical edge} of $g$ corresponds to a shared endpoint of two pieces that are both touched first for radii that are strictly smaller than the radius of the shared endpoint.
Each endpoint shared by two pieces is also the shared endpoint of two segments in the same weighted Voronoi tessellation, and the corresponding edge in $\Delaunay{A_0', A_1', \ldots, A_s'}$ is critical iff the edge in the corresponding weighted Delaunay mosaic is critical.

\smallskip
In contrast to the critical edges, there are two cases for the critical vertices of $g$.
First, the lune that touches the corresponding piece does so at a point interior to a circular arc in the boundary of the lune.
Hence, there is a disk that touches $\Rspace^1$ at the same point, this point is in the interior of a segment, and the corresponding vertex in the weighted Delaunay mosaic has the same assigned radius.
Second, the lune that touches the piece does so at a tip in the boundary of the lune.
The existence of critical vertices of the second type is the reason why the following extension of Lemma~\ref{lem:weighted_radii} is not immediate.
\begin{lemma}
  \label{lem:chromatic_weighted_radii}
  Let $s$ be a non-negative integer, $A_0, A_1, \ldots, A_s$ stationary Poisson point processes in $\Rspace^2$, possibly with different intensities, $\Rspace^1 \hookrightarrow \Rspace^2$ a fixed line, $A_0', A_1', \ldots, A_s'$ the weighted points obtained by projecting each $A_i$ orthogonally to $\Rspace^1$, $\Omega = [0,1] \subseteq \Rspace^1$, and consider the radius function on the weighted lunar Delaunay mosaic dual to the overlay of the weighted Voronoi tessellations of the $A_i'$.
  Then the expected sum of radii of the critical vertices and edges whose centers belong to $\Omega$ is at most $O(1)$.
\end{lemma}
\begin{proof}
  By Lemma~\ref{lem:weighted_radii}, the sum of radii assigned to the critical vertices and edges of the $s+1$ weighted Delaunay mosaics, whose center belongs to $\Omega$, is at most $O(1)$ in expectation, for each $0 \leq i \leq s$.
  This includes every critical edge of $g$ and every critical vertex whose corresponding piece is first touched by a circular arc in the boundary of the lune.
  It therefore suffices to consider the critical vertices of $g$ whose corresponding pieces in the weighted Voronoi tessellation are touched first by tips of lunes.

  \smallskip
  Consider such a piece and let $\aaa = (a_0, a_1, \ldots, a_s)$ be the colorful $s$-simplex such that the piece is the common intersection of the segments $\sigma_i$, for $0 \leq i \leq s$, in which $\sigma_i$ is the intersection of the domain of $a_i$ in $\Voronoi{A_i}$ with $\Rspace^1$.
  Let $x\in \Omega$ be the point of the piece touched first by the lune of $\aaa$, and assume $\Edist{x}{a_0} \geq \Edist{x}{a_1} \geq \ldots \geq \Edist{x}{a_s}$.
  Since $x$ is a tip, we have $\Edist{x}{a_0} = \Edist{x}{a_1}$.
  Then $a_0'$ and $a_1'$ are vertices of the (un-colored) weighted Delaunay mosaic of $A_0' \cup A_1'$, and the edge connecting them is a critical edge that shares the radius with the $\aaa$.
  Indeed, since $\aaa$ is critical, $x$ lies in the interiors of $\sigma_0$ and $\sigma_1$, and the circle centered at $x$ passing through $a_0$ and $a_1$ therefore encloses no points of $A_0$ and $A_1$. 
  By Lemma~\ref{lem:weighted_radii}, the sum of radii of all edges in $\Delaunay{A'_i \cup A'_j}$ is $O(1)$ for all choices of colors $0 \leq i < j \leq s$, which finishes the proof.
\end{proof}

\section{Main Theorem}
\label{sec:4}

In this section, we formally state the main result of this paper, which asserts the existence of an asymptotic constant for random lunar EMSTs.
The proof strategy is an adaption of a general probabilistic method developed among others by Steele~\cite{Ste97} and Yukich~\cite{Yuk98}.

\subsection{Statement and Proof Strategy}
\label{sec:4.1}

Fix a non-negative integer constant, $s$.
Given a finite set $A \subseteq \Rspace^2$ and a coloring $\chi \colon A \to [s]$, we write $A_i = \chi^{-1} (i)$, for $0 \leq i \leq s$, and $\Cost{}{A_0, A_1, \ldots, A_s}$ for the cost of the $s$-lunar EMST of $A$, as defined in \eqref{eqn:cost}.
The main result of this paper is the existence of a constant that describes the expected cost for random points with random coloring in the limit.
\begin{theorem}
  \label{thm:asymptotic_constant}
  Let $s \geq 0$ be an integer, $A$ a set of $n$ points sampled uniformly at random in $[0,1]^2$, and $\chi \colon A \to [s]$ a random $(s+1)$-coloring.
  Then there exists a constant, $c_s$, such that the expected cost of the $s$-lunar EMST of $\chi$ is $c_s \sqrt{n}$, in the limit, when $n$ goes to infinity; that is:
  \begin{align}
    \lim\nolimits_{n \to \infty} \Expect{A, \chi}{\Cost{}{A_0, A_1, \ldots, A_s}} / \sqrt{n} &= c_s .
    \label{eqn:main_theorem}
  \end{align}
\end{theorem}
The proof strategy is not new and applies to general \emph{Euclidean functionals}, which are maps from finite subsets of Euclidean space to the non-negative real numbers that are normalized, homogeneous, and translation invariant; see Conditions~I, II, III below.
We begin by formulating the cost of a lunar EMST so it fits this framework.
For a finite set $A\subseteq\Rspace^2$, we write $\Cost{s}{A}$ for the average cost over all possible $(s+1)$-colorings of $A$.
Note that the averaging over all $(s+1)$-colorings is implicit in this notation, so we can shorten \eqref{eqn:main_theorem} by substituting $\Expect{A}{\Cost{s}{A}}$ for $\Expect{A, \chi}{\Cost{}{A_0, A_1, \ldots, A_s}}$.
In the probabilistic setting, we consider locally finite point sets $A \subseteq \Rspace^2$---often sampled from a stationary Poisson point process with density $n$---choose the points within some compact region---usually $Q = [0,1]^2$---and write $\Cost{}{A_0, A_1, \ldots, A_s; Q}$ for the cost of the lunar EMST of $A_i \cap Q$, for $0 \leq i \leq s$, and $\Cost{s}{A; Q}$ for the average cost over all $(s+1)$-colorings of $A \cap Q$.

\smallskip
The proof consists of a sequence of technical steps, and it pays off to preview these steps and discuss how they relate to each other as well as how they contribute to the final goal.
To begin, we note that the cost function satisfies the three conditions of a Euclidean functional.
As usual, multiplication with a scalar means scaling, and addition of a vector means translating:
\medskip \begin{enumerate}[II]
  \item[I] \textbf{normalization:} $\Cost{s}{\emptyset; Q} = 0$;
  \smallskip
  \item[II] \textbf{homogeneity:} $\Cost{s}{\lambda A; \lambda Q} = \lambda \cdot \Cost{s}{A; Q}$ for every real $\lambda > 0$;
  \smallskip
  \item[III] \textbf{translation invariance:} $\Cost{s}{A+y; Q+y} = \Cost{s}{A; Q}$ for every $y \in \Rspace^2$.
\end{enumerate}
\medskip
It should be clear that Conditions~I, II, and III are indeed satisfied by the average cost of the lunar EMST over all $(s+1)$-colorings of the points.
The remaining conditions require special properties of the cost function, including its sub-additivity.
For partially technical reasons, the proofs of these properties work on slightly different versions of the tree, which we now define together with their costs.

\smallskip
Let $A \subseteq \Rspace^2$ be locally finite, $\chi \colon A \to [s]$ an $(s+1)$-coloring, and recall that $\Voronoi{A_0, A_1, \ldots, A_s}$ is the overlay of the $s+1$ Voronoi tessellations.
Fixing $Q$ to a compact subset of $\Rspace^2$ with connected interior---usually a square or rectangle---we \emph{restrict} $\Voronoi{A_0,A_1,\ldots,A_s}$ by intersecting its cells with the interior, $\interior{Q}$.
Correspondingly, the \emph{lunar Delaunay mosaic} of $\chi$ \emph{restricted} to $Q$ is the dual of this restricted tessellation; that is: the collection of cells in $\Delaunay{A_0, A_1, \ldots,
A_s}$ whose corresponding cells in $\Voronoi{A_0,A_1,\ldots,A_s}$ have a non-empty intersection with $\interior{Q}$.
The matching radius function maps each Delaunay cell to the minmax distance of any point in the corresponding cell in the overlay restricted to $\interior{Q}$; that is: $\inf_x f_{\max} (x)$, in which $f_{\max} (x) = \max_{0 \leq i \leq s} \min_{a \in A_i} \Edist{x}{a}$.
The sublevel sets of this radius function capture the topology of the union of all the lunes of $A$ restricted to $\interior{Q}$.
\begin{definition}
  \label{dfn:restricted_lunar_EMST}
  Let $A \subseteq \Rspace^2$ be locally finite, $\chi \colon A \to [s]$ an $(s+1)$-coloring, and $Q \subseteq \Rspace^2$ compact with connected interior.
  The \emph{lunar EMST} of $\chi$ \emph{restricted} to $Q$ is the spanning tree within the lunar Delaunay mosaic restricted to $Q$ that minimizes the cost, denoted $\Costrestr{}{A_0, A_1, \dots, A_s; Q}$.
\end{definition}
See the middle panel of Figure~\ref{fig:threeTrees} for an example of a restricted lunar EMST.
To introduce the second variant, add the complement of $\interior{Q}$ as an extra domain to the restricted overlay of Voronoi tessellations, and define the \emph{rooted lunar Delaunay mosaic} of $\chi$ as its dual.
In other words, we add a single node to the restricted lunar Delaunay mosaic, called the \emph{root}, assign to it the radius $0$, and connect it with the duals of all cells in $\Voronoi{A_0, A_1, \ldots, A_s}$ that have a non-empty intersection with the boundary of $Q$.
We call each such arc an \emph{anchor}, and set its lunar radius to the radius at which the lune of the domain first touches the boundary of $Q$---and its cost, as always, to twice the radius.
\begin{figure}[hbt]
    \centering
    \includegraphics[width=\linewidth]{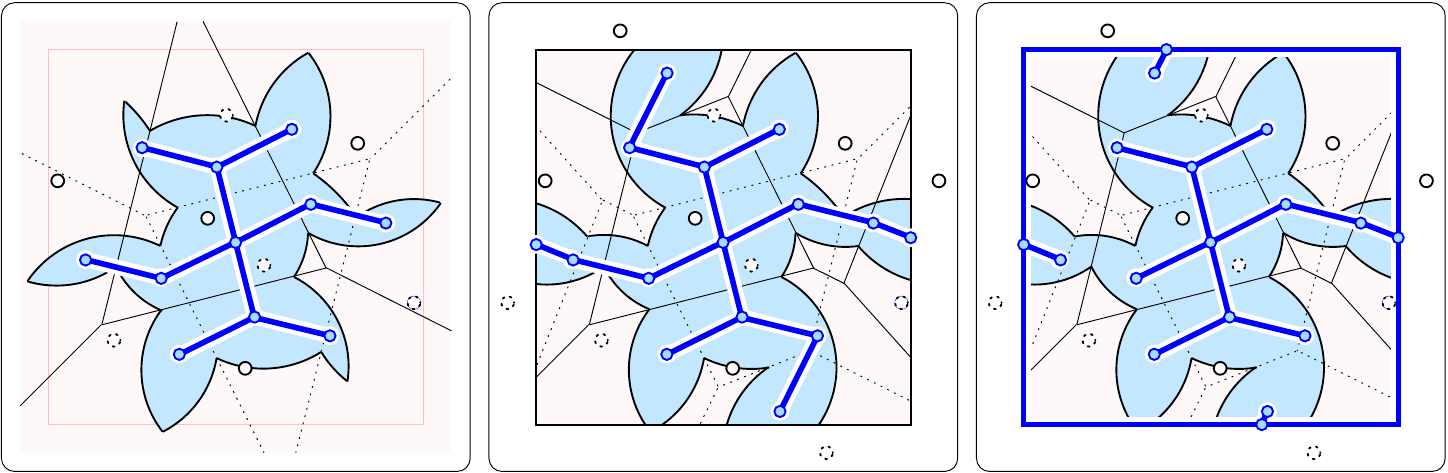}
    \caption{\footnotesize \emph{Left:} the lunar EMST of points in $[0,1]^2$; the lunes may extend outside $[0,1]^2$.
    \emph{Middle:} the lunar EMST of points in $\Rspace^2$ \emph{restricted} to $[0,1]^2$; there may be points outside $[0,1]^2$ that contribute to domains and lunes partially or entirely inside $[0,1]^2$.
    \emph{Right:} the \emph{rooted} lunar EMST of points in $\Rspace^2$ restricted to $[0,1]^2$; the boundary of $[0,1]^2$ acts as an extra vertex of the tree.}
    \label{fig:threeTrees}
\end{figure}
\begin{definition}
    \label{dfn:rooted_lunar_EMST}
    Let $A \subseteq \Rspace^2$ be locally finite, $\chi \colon A \to [s]$ an $(s+1)$-coloring, and $Q \subseteq \Rspace^2$ compact with connected interior.
    The \emph{rooted lunar EMST} of $\chi$ is the spanning tree within the rooted lunar Delaunay mosaic---in which the root itself may or may not be spanned---that minimizes the cost, and we denote this minimum cost by $\Costrooted{}{A_0, A_1, \dots, A_s; Q}$.%
    \footnote{
        The definition of the rooted lunar EMST is inspired by the notion of relative homology, in the sense that we do not insist that the root is part of the spanning tree.
        Nevertheless, if the spanning tree contains an anchor connecting to the root, the cost of the anchor is part of the overall cost.
        Since the cost of the root is $0$, counting or not counting it makes no difference to the overall cost, but the possibility of not including the root implies that the restricted lunar EMST is also a rooted spanning tree, albeit not necessarily the one with minimum cost.
    }
\end{definition}
See the right panel of Figure~\ref{fig:threeTrees} for an example of a rooted lunar EMST.
As before, we work with the average cost over all $(s+1)$-colorings of $A$, but note that it suffices to average over all $(s+1)$-colorings of the vertices in the lunar Delaunay mosaic restricted to $Q$.
We therefore write
\begin{align}
  \Cost{s}{A; Q} &= \Expect{\chi}{\Cost{}{A_0, A_1, \ldots, A_s; Q}} ; \\
  \Costrestr{s}{A; Q} &= \Expect{\chi}{\Costrestr{}{A_0, A_1, \ldots, A_s; Q}} ; \\
  \Costrooted{s}{A; Q} &= \Expect{\chi}{\Costrooted{}{A_0, A_1, \ldots, A_s; Q}} . 
\end{align}
The average costs of the restricted and rooted lunar EMSTs again fulfill Conditions I, II, III and are therefore Euclidean functionals.
To formulate the remaining five conditions needed to prove the existence of the asymptotic constant, we let $M$ be a positive integer, and $Q_1, Q_2, \ldots, Q_{M^2}$ the copies of $[0, \frac{1}{M}]^2$ that decompose $Q=[0,1]^2$ into $M^2$ smaller squares.
The conditions require the existence of constants $C_s^{\rm sub}$, $C_s^{\rm one}$, $C_s^{\rm sup}$, $C_s^{\rm two}$, $C_s^{\rm smth}$ such that
\medskip \begin{enumerate}[II]
  \item[IV] \textbf{sub-additivity:}
  $\Costrestr{s}{A; Q} \leq \sum\nolimits_{1 \leq i \leq M^2} \Costrestr{s}{A; Q_i} + C_s^{\rm sub} M$, for all $M \geq 1$ and all locally finite $A \subseteq \Rspace^2$;
  \smallskip
  \item[V] \textbf{mean closeness one:} $\Expect{A}{| \Cost{s}{A; Q} - \Costrestr{s}{A; Q} |} \leq C_s^{\rm one}$, in which the expectation is over a stationary Poisson point process in $\Rspace^2$;
  \smallskip
  \item[VI] \textbf{super-additivity:}
  $\Costrooted{s}{A; Q} \geq \sum\nolimits_{1 \leq i \leq M^2} \Costrooted{s}{A; Q_i} - C_s^{\rm sup} M$, for all $M \geq 1$ and all locally finite $A \subseteq \Rspace^2$;
  \smallskip
  \item[VII] \textbf{mean closeness two:} $\Expect{A}{| \Costrestr{s}{A; Q} - \Costrooted{s}{A; Q} |} \leq C_s^{\rm two}$, in which the expectation is over a stationary Poisson point process in $\Rspace^2$;
  \smallskip
  \item[VIII] \textbf{add-one bound:} $\Expect{A,a}{| \Cost{s}{A; Q} - \Cost{s}{A'; Q} |} \leq {C_s^{\rm smth}}/{\sqrt{n}}$, in which $n = \card{A}$, $A' = A \setminus \{a\}$, and $a \in A$ is randomly selected among
  uniformly sampled points $A \subseteq Q$.
\end{enumerate} \medskip
The difficulty in proving Condition~IV directly for $\Cost{s}{A; Q}$ motivates the introduction of the lunar EMST restricted to $Q$.
We need the cost of this EMST to differ by at most a constant from that of the unrestricted EMST, in expectation, which is Condition~V.
A further complication arises because the cost function is not monotonic; that is: it is possible that adding a point leads to a decrease of the cost.
It shares this lack of monotonicity with the standard Euclidean minimum spanning tree, for which it motivated the study of Steiner minimal trees, as initiated by Gilbert and Pollack~\cite{GiPo68}.
This is the reason we introduce the rooted lunar EMST, in which the boundary of $[0,1]^2$ sometimes acts as a node in the tree; see the right panel in Figure~\ref{fig:threeTrees} for an illustration.
We require that it has the opposite tendency upon subdivision: it is super-additive as formalized in Condition~VI.
Since the rooted and un-rooted lunar trees are different, we again need that the expected costs are not very different, which is Condition~VII.

\smallskip
The proofs of IV and VI do not depend on randomness, neither of the points nor the coloring, while the proofs of V and VII need the randomness of the points and of the coloring.
However, we establish V and VII for a stationary Poisson point process as opposed to points sampled from the uniform distribution, which is the reason we need Condition~VIII---sometimes referred to as the \emph{combinatorial smoothness} of the cost function---to make up for the difference to points drawn from a uniform distribution.
Similar to V and VII, the proof of VIII needs the randomness of the points as well as the coloring.

\smallskip
According to Yukich's book \cite[Theorem~5.2]{Yuk98}, Conditions~I to VIII imply the existence of the asymptotic constant, $c_s$, for the expected value of the Euclidean functional $\Cost{s}{A}$ of $n$ points sampled uniformly at random in $[0,1]^2$.
In other words, they suffice to prove the main result of this paper, which is Theorem~\ref{thm:asymptotic_constant}.
We describe the establishment of Conditions~IV to VIII in the subsequent subsections, one for each condition.

\begin{remark}
  The eight conditions and their proofs generalize to the setting in which the colors have different probabilities.
  Let $\ppp = (p_0, p_1, \ldots, p_s)$ be a vector of probabilities, with $\sum_i p_i = 1$, and suppose each point is colored independently and receives color $i$ with probability $p_i$.
  If $A$ is a stationary Poisson point process with intensity $n$, then $A_i = \chi^{-1} (i)$ is a stationary Poisson point process with intensity $p_i n$, for each $0 \leq i \leq s$, and if $A$ is sampled randomly from a uniform distribution, then so is $A_i$, for each $0 \leq i \leq s$.
  In the following, whenever we say that $\chi$ is a \emph{random coloring} of a point process $A$, we mean it in this sense, for any fixed $\ppp$.
  Hence, we can define the Euclidean functionals $\Cost{s}{A; Q}$, $\Costrestr{s}{A; Q}$, $\Costrooted{s}{A; Q}$ as averages over all colorings weighted according to the probabilities $\ppp$, and Theorem~\ref{thm:asymptotic_constant} generalizes to a random set of points colored randomly according to $\ppp$:
  \end{remark}
  \begin{corollary}
    \label{cor:asymptotic_constant}
    Let $s \geq 0$ be an integer, $\ppp = (p_0, p_1, \ldots, p_s)$ a vector of probabilities that add to $1$, $A$ a set of $n$ points sampled uniformly at random in $[0,1]^2$, and $\chi \colon A \to [s]$ a random $(s+1)$-coloring in which each point receives color $i$ with probability $p_i$, for $0 \leq i \leq s$.
    Then there exists a constant, $c_s^\ppp$, such that the expected cost of the lunar EMST of $\chi$ is $c_s^\ppp \sqrt{n}$, in the limit, when $n$ goes to infinity.
  \end{corollary}
  See Figure~\ref{fig:two-colors} for an example, which shows the estimated asymptotic constants for two colors and probabilities $p_0 = \lambda$ and $p_1 = 1 - \lambda$.
  In the proofs of Conditions~IV to VIII, $\ppp$ is fixed and dropped from the notation.
    
\begin{figure}[hbt]
    \centering
    \includegraphics[width=.535\linewidth]{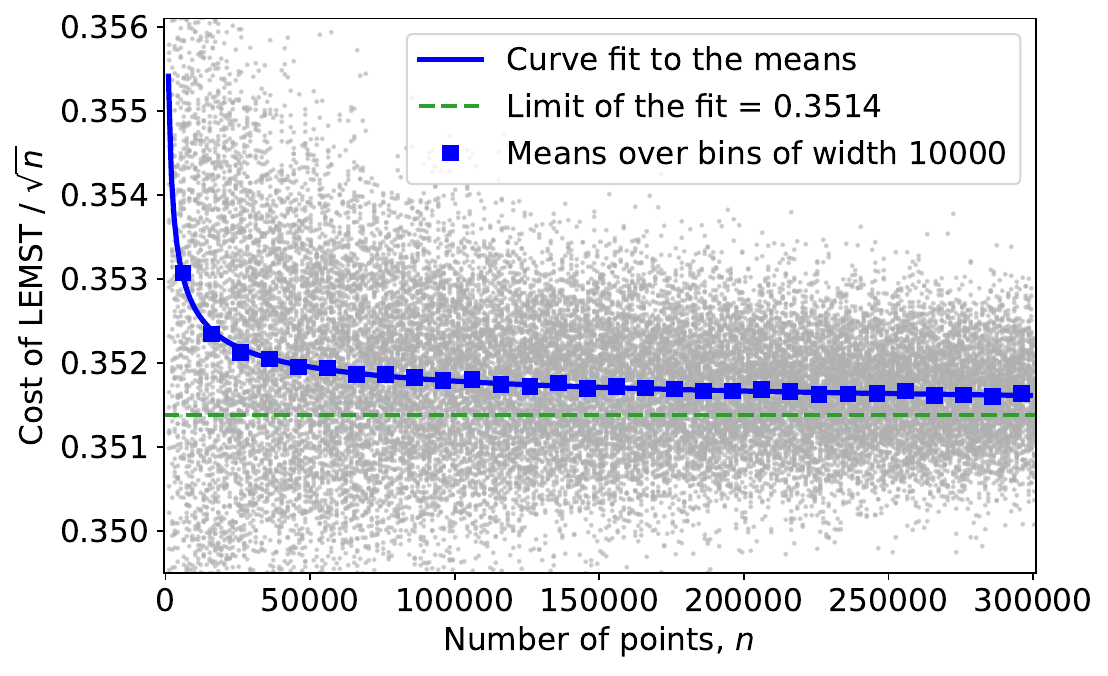}%
    \includegraphics[width=.46\linewidth]{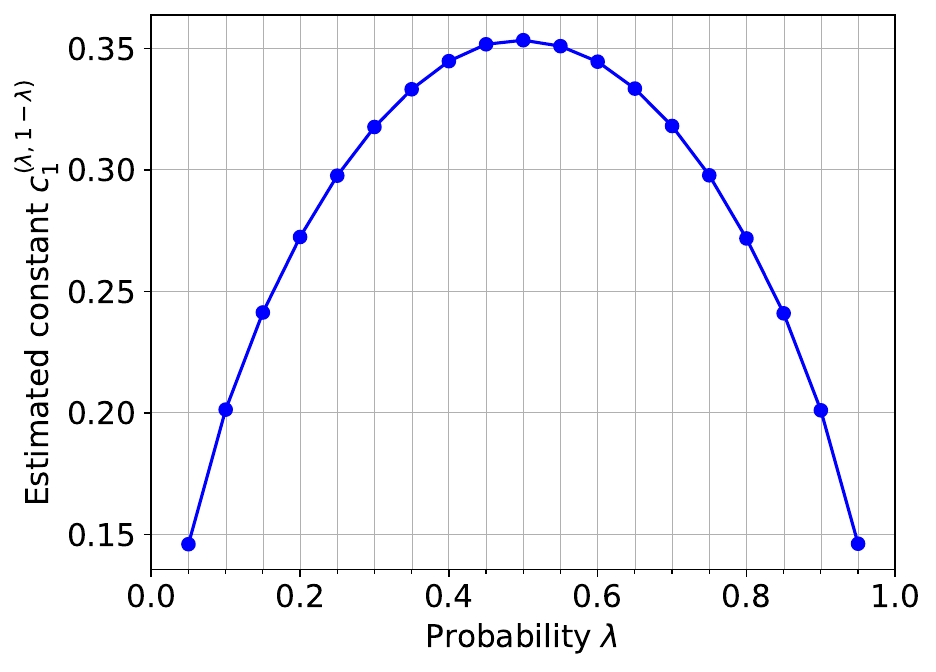}
    \caption{\footnotesize
    Estimates of the asymptotic constant of the lunar EMST cost for points sampled from $[0,1]^2$. 
    \emph{Left:} the estimates of the asymptotic constant, $c_s$, with $s=1$. 
    For small values of $n$, the boundary effects skew the estimates higher than the limit of a fitted curve, which estimates $c_s$ around $0.351\ldots$, consistent with the aforementioned estimates in \cite{DERS25}.
    \emph{Right:} the estimated asymptotic constants, $c_s^\ppp$, with $s = 1$ and probability vectors $\ppp = (\lambda, 1- \lambda)$ for $\lambda \in [0,1]$. 
    The computations are performed with $n=2000$ points sampled uniformly at random from $[0,1]^2$.  
    }
    \label{fig:two-colors}
\end{figure}

\subsection{The Cost of the Restricted Lunar EMST is Sub-additive}
\label{sec:4.2}

We need a technical lemma to prove that Condition~IV is satisfied.
The setting consists an $(s+1)$-coloring of a locally finite set $A \subseteq \Rspace^2$ and the squares $Q' = [0,1]^2$ and $Q'' = [1,2] \times [0,1]$.
We compare the costs of lunar EMSTs restricted to $Q'$, $Q''$, and $Q = Q' \cup Q''$.
\begin{lemma}\label{lem:sub-additivity}
  Let $A \subseteq \Rspace^2$ be locally finite, $s \geq 0$ an integer, and $\chi \colon A \to [s]$ an \Skip{random} $(s+1)$-coloring.
  The costs of the lunar EMSTs of $A$ restricted to the two unit squares and the rectangle satisfy
  \begin{align}
    \Costrestr{s}{A_0,A_1,\ldots,A_s; Q} &\leq \Costrestr{s}{A_0,A_1,\ldots,A_s; Q'} + \Costrestr{s}{A_0,A_1,\ldots,A_s; Q''} + 2 \sqrt{2} .
  \end{align}
\end{lemma}
\begin{proof}
  Given the lunar EMSTs of $\chi$ restricted to $Q'$ and $Q''$, we construct a spanning tree within $Q = Q' \cup Q''$ whose cost exceeds the sum of the costs of the two lunar EMSTs by at most $2 \sqrt{2}$.
  The cost of the lunar EMST restricted to $Q$ is at most the cost of that tree, which will imply the inequality.

  \smallskip
  As first step, we take the union of the two restricted lunar EMSTs, noting that each domain that overlaps both squares is represented by two nodes that need to be merged into one---for the cost, equivalently, an edge is added with cost equal to the higher of the nodes' costs.
  To see how this is done, let $\aaa$ be a colorful $s$-simplex such that $\dom{\aaa} \cap Q' \neq \emptyset$ and $\dom{\aaa} \cap Q'' \neq \emptyset$, as in Figure~\ref{fig:merging}.
  Let $u', u''$ be the corresponding two nodes, which do not necessarily have the same lunar radius, as each is born at the moment the lune has a non-empty intersection with the corresponding restricted domain.

  \smallskip
  Assuming $\Lambda (u') \leq \Lambda (u'')$, we have $\Lambda (u) = \Lambda (u')$ for the merged node $u$.
  Letting $\omega'$ and $\omega''$ be the nodes with smallest costs in the two lunar EMSTs, this operation increases the combined cost by at most $\Lambda (u'') - \max \{\Lambda (\omega'), \Lambda (\omega'')\}$.
  Since both squares have diameter $\sqrt{2}$, we have $\Lambda (u'') \leq \max \{\Lambda (\omega'), \Lambda (\omega'')\} + 2 \sqrt{2}$.
  It follows that at this stage of the construction, the increment of the cost of the combined tree beyond the sum of the costs of the two lunar EMSTs is at most $2 \sqrt{2}$.

  \smallskip
  Let $\bbb$ be another colorful $s$-simplex with $\dom{\bbb} \cap Q' \neq \emptyset$ and $\dom{\bbb} \cap Q'' \neq \emptyset$.
  Letting $v'$ and $v''$ be the corresponding nodes, we again merge them into one, and assuming $\Lambda (v') \leq \Lambda (v'')$, the cost of the new node is $\Lambda (v) = \Lambda (v')$.
  But now the operation creates a cycle, which we remove by deleting an arc $e$ incident to $v''$.
  By construction, $\Lambda (e) \geq \Lambda (v'')$ so the increment to the cost is $\Lambda (v'') - \Lambda (e)$, which is non-positive.
  Similarly, we merge the nodes that correspond to the remaining domains overlapping both squares, each time deleting an arc.
  Since none of these operations increases the cost of the tree, the inequality established after merging $u'$ and $u''$ still holds.
\end{proof}

\begin{figure}[hbt]
    \centering
    \vspace{0.0in}
    \resizebox{!}{1.2in}{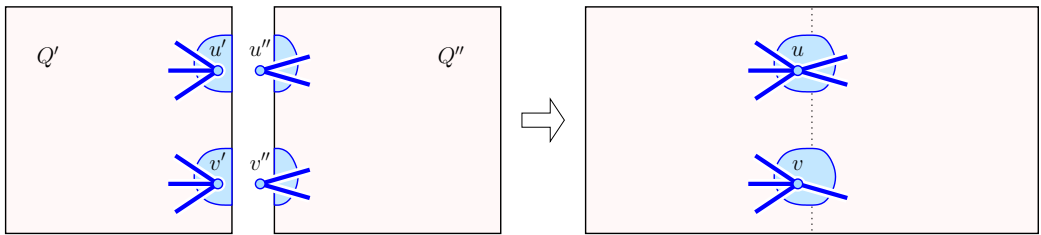}
    \caption{\footnotesize Merging the lunar EMSTs in adjacent copies of the unit square.
    Except at the first time, each merger leads to the removal of an incident arc.}
    \label{fig:merging}
\end{figure}

We use Lemma~\ref{lem:sub-additivity} to prove that Condition~IV is satisfied.
Assume $M = 2^k$.
To consolidate the $M$-by-$M$ grid to a single square, we take $k$ stages, in each merging the squares in groups of four to generate a coarser grid.
Consider the first step in the first stage, which merges
\begin{align}
  [0,\tfrac{1}{M}] \times [0,\tfrac{1}{M}], ~~~~~
  [\tfrac{1}{M},\tfrac{2}{M}] \times [0,\tfrac{1}{M}],~~~~
  [0,\tfrac{1}{M}] \times [\tfrac{1}{M},\tfrac{2}{M}], ~~~~
  [\tfrac{1}{M},\tfrac{2}{M}] \times [\tfrac{1}{M},\tfrac{2}{M}]
\end{align}
into $[0,\tfrac{2}{M}]^2$.
By Lemma~\ref{lem:sub-additivity}, merging the first two squares increases the cost by at most ${2 \sqrt{2}} \cdot \frac{1}{M}$, merging the second two squares increases the cost by at most the same amount, and merging the two rectangles increases the cost by at most twice that amount.
There are $M^2 / 4$ steps in the first stage, so the cost increment in this stage is at most
\begin{align}
  \tfrac{M^2}{4} \left( 2 \sqrt{2} \cdot \tfrac{1}{M} + 
  2 \sqrt{2} \cdot \tfrac{1}{M} + 2 \sqrt{2} \cdot \tfrac{2}{M} \right) = 2 \sqrt{2} M .
\end{align}
The increment in the second stage is at most half this amount,
etc., so the total cost increment is at most $2 \sqrt{2} M \cdot (1 + \sfrac{1}{2} + \ldots + \sfrac{1}{2^k}) < 4 \sqrt{2} M$. 
This proves Condition~IV for $C_s^{\rm sub} = 4 \sqrt{2}$ if $M$ is a power of $2$, and for $C_s^{\rm sub} = 8 \sqrt{2}$ if $M$ is any positive integer.

\subsection{Mean Closeness of Restricted and Un-restricted Lunar EMSTs}
\label{sec:4.3}

To prove that Condition~V is satisfied, we use the fact that only the critical nodes and arcs of the lunar radius function are relevant for the cost.
By Lemma~\ref{lem:critical_and_non-critical_cells}, these nodes and arcs are characterized by the smallest enclosing stacks of the corresponding points.
Considering a Poisson point process, $A \subseteq \Rspace^2$, as well as its restriction to the unit square, $A \cap [0,1]^2$, we focus on the sizes of the smallest enclosing stacks defined by these points.
If such a stack lies completely inside $[0,1]^2$, then it corresponds to a node or arc in both, the restricted and the unrestricted lunar EMSTs.
The nodes and arcs in the symmetric difference thus correspond to smallest enclosing stacks that either cross the boundary of $[0,1]^2$ or lie completely outside $[0,1]^2$.
\begin{lemma}
  \label{lem:mean_closeness_one}
  Let $A$ be a stationary Poisson point process with positive intensity in $\Rspace^2$, $s \geq 0$ an integer, and $Q = [0,1]^2$.
  Then the expected absolute difference between the average costs of the lunar EMST of $A \cap [0,1]^2$ and the lunar EMST of $A$ restricted to $Q$ is at most some constant; that is: $\Expect{A}{| \Cost{s}{A; Q} - \Costrestr{s}{A; Q} |} \leq C_s^{\rm one}$.
\end{lemma}
\begin{proof}
  Let $\chi \colon A \to [s]$ be a random $(s+1)$-coloring.
  We write $A_i = \chi^{-1} (i)$, as usual, and introduce $B = A \cap Q$ and $B_i = A_i \cap Q$ for $0 \leq i \leq s$.
  The lunar radius function, $\Lambda_B$, is on an un-restricted lunar Delaunay mosaic, while $\Lambda_A$ is on a restricted such mosaic.
  We distinguish between three kinds of nodes and arcs that are critical for one of the two lunar radius functions but not the other:
    \smallskip \begin{enumerate}[(i)]
        \item critical nodes and arcs of $\Lambda_A$ whose smallest enclosing circles cross $\partial Q$ and have centers inside $Q$; 
        
        \item critical nodes and arcs of $\Lambda_B$ whose smallest enclosing circles cross $\partial Q$;
        
        \item critical nodes and arcs of $\Lambda_A$ whose smallest enclosing circles have centers outside $Q$.
    \end{enumerate} \smallskip
    Since the goal are upper bounds, we will enlarge the sets as convenient.
    For example, we will ignore the constraint that the circles in (i) have their centers inside $Q$ and consider all such circles that cross the boundary of $Q$.
    The remainder of this proof considers the three types in sequence, and for each type shows that the sum of lunar radii is bounded by a constant independent of $n$.
    
    \smallskip
    We use Lemma~\ref{lem:crossing_the_boundary} to get the bounds for the critical nodes and arcs of type (i).
    By Lemma~\ref{lem:critical_and_non-critical_cells}, a node is critical only if its smallest enclosing stack is empty. Assuming general position of the points, the largest circle of this stack—which is also the smallest enclosing circle of the node—passes through either two or three of the s + 1 points.
    If two, and these two points belong to $A_i$ and $A_j$, then they define a critical edge in $\Delaunay{A_i \cup A_j}$.
    If three, and these three points belong to $A_i$, $A_j$, and $A_k$, then they define a critical triangle in $\Delaunay{A_i \cup A_j \cup A_k}$.
    By Lemma~\ref{lem:crossing_the_boundary}, the expected sum of the radii of the circles that cross $\partial Q$ is therefore at most some constant times $s^3$, which itself is a constant.
    %
    Again by Lemma~\ref{lem:critical_and_non-critical_cells}, an arc of the lunar radius function is critical only if the smallest enclosing stack of the corresponding $s + 2$ points is empty. Assume $A_j$ contributes two to the $s+2$ points while every other set contributes only one point. Since the smallest enclosing stack of the arc has to be strictly larger than those of its two nodes, both points contributed by $A_j$ lie on the largest circle in the stack.
    If they are the only two points on the circle, then they span a critical edge in $\Delaunay{A_j}$.
    Otherwise, there is a third point on the circle---say a point of $A_k$---in which case the three points span a critical triangle in $\Delaunay{A_j \cup A_k}$.
    The expected sum of the radii of the circles that cross $\partial Q$ is therefore again at most $O(1)$.
    
    \smallskip
    With the same argument, we use Lemma~\ref{lem:crossing_the_boundary} together with Lemma~\ref{lem:critical_and_non-critical_cells} to get the bounds for the critical nodes and arcs of type (ii).
    Indeed, the centers of the smallest enclosing circles of points in $B$ lie inside $Q$ by construction, so Lemma~\ref{lem:crossing_the_boundary} applies without modification.
    The expected sums of radii of the circles that correspond to critical nodes and arcs of $\Lambda_B$ that cross $\partial Q$ are therefore at most $O(1)$.

    \smallskip
    The nodes and arcs of type (iii) correspond to critical vertices and edges of the overlay of $s+1$ weighted Voronoi tessellations along the boundary of the unit square.
    We assume that $\chi$ is a random coloring, so Lemma~\ref{lem:chromatic_weighted_radii} implies that the expected sum of their radii is at most $O(1)$, given that the centers of the corresponding cells in the weighted lunar Delaunay mosaic lie on the boundary of $[0,1]^2$.
    This is always the case for edges, but might not be the case for vertices, when the interior of a domain in the overlay of Voronoi tessellations contains a corner of the unit square. 
    There are at most four such cases, one for each corner of $[0,1]^2$.
    Each such node used in the minimum spanning tree will be paired with an arc that is critical in $\Lambda_A$ but not critical in $\Lambda_B$, and whose radius differs at most by $\sqrt{2}$.
    So overall they contribute at most a constant to the difference between the costs of the two lunar trees.
\end{proof}

\subsection{The Cost of the Rooted Lunar EMST is Super-additive}
\label{sec:4.4}

We need a technical lemma to prove that Condition~VI is satisfied.
Similar to the restricted lunar EMST, the rooted lunar EMST is obtained by growing the lunes inside their domains intersected with $Q = [0,1]^2$, and adding nodes and arcs with lunar radii accordingly.
The important difference is that we start with the root of cost~$0$, which represents $\partial Q$.
Whenever a lune restricted to its domain intersected with $Q$ touches $\partial Q$, we add an arc to the root, provided this is the first time the component to which the lune belongs touches $\partial Q$. 
We call this arc an anchor, and to be consistent with the growth model of lunes, we set its cost to twice the radius at which the touching happens.
$\partial Q$ acts as a single point, so whenever two components touch $\partial Q$, they are considered to belong to the same component.
As explained in the footnote of Definition~\ref{dfn:rooted_lunar_EMST}, it is also possible that no anchor is used in the construction.
The \emph{cost} of the rooted tree is defined as that of the other trees: the sum of costs of the arcs minus the costs of the nodes plus the cost of the minimum cost node.
Like before $Q' = [0,1]^2$, $Q'' = [1,2] \times [0,1]$, and $Q = Q' \cup Q''$.
\begin{lemma}
  \label{lem:super-additivity}
  Let $A \subseteq \Rspace^2$ be locally finite, $s \geq 0$ an integer, and $\chi \colon A \to [s]$ an $(s+1)$-coloring.
  Then the costs of the rooted lunar EMSTs restricted to the two unit squares and the rectangle satisfy
  \begin{align}
    \Costrooted{s}{A_0, A_1, \ldots, A_s; Q} \geq \Costrooted{s}{A_0, A_1, \ldots, A_s; Q'} + \Costrooted{s}{A_0, A_1, \ldots, A_s; Q''} .
  \end{align}
\end{lemma}
\begin{proof}
  We show that the rooted lunar EMST of $\chi$ restricted to $Q$ can be split into two rooted trees, one restricted to $Q'$ and the other to $Q''$, such that their combined cost is at most $\Costrooted{s}{A; Q}$.
  The rooted lunar EMSTs restricted to the two squares do not have larger cost, which implies the claimed inequality.
    
  \smallskip
  Whenever we split a node of a tree into two---with each neighbor of the original node connected to either one of the two copies---we produce two (connected) trees; see Figure~\ref{fig:splitting}.
  Assuming there are $k$ colorful $s$-simplices whose domains overlap both squares, we split $k+1$ nodes, including the root that represents $\partial Q$.
  We therefore produce $k+2$ trees, so we need $k$ anchors to connect them into two trees, one within each square.
  It thus suffices to add one anchor for each of the $k$ mentioned $s$-simplices.
    
  \smallskip
  Let $\aaa$ be such a colorful $s$-simplex, so $\dom{\aaa} \cap Q' \neq \emptyset$ and $\dom{\aaa} \cap Q'' \neq \emptyset$.
  Let $u$ be the corresponding node in the rooted lunar EMST of $Q$, which we split into two nodes, $u'$ and $u''$.
  Assume without loss of generality that $\Lambda (u) = \Lambda (u') \leq \Lambda (u'')$, and in this case add an anchor $e$ connecting $u'$ to $\partial Q'$ or $u''$ to $\partial Q''$.
  In addition, we connect each neighboring node, $v$, of $u$ to either $u'$ or $u''$ depending on whether its domain belongs to $Q'$ or $Q''$.
  There is also the case in which $v$ is of the same kind as $u$ and thus splits into $v'$ and $v''$, in which $v'$ is possibly connected to $u'$ and $v''$ is possibly connected to $u''$.
  However, at most one of these two arcs is needed, since either $v'$ or $v''$ is connected by an anchor to its square's root, just like $u'$ or $u''$ is, and the cost of this one arc is the cost of the old arc connecting $u$ to $v$.

  \smallskip
  The increment in cost caused by the splitting of $u$ is at most $\Lambda (e) - \Lambda (u'')$, which is non-positive because the lune that first grows within $\dom{\aaa} \cap Q'$ touches the boundary of its square before it also grows inside $\dom{\aaa} \cap Q''$, so $\Lambda (e) \leq \Lambda (u'')$.
  This is true for all nodes that are split, which implies the claimed inequality.
\end{proof}

\begin{figure}[hbt]
  \centering
  \vspace{0.0in}
  \resizebox{!}{1.2in}{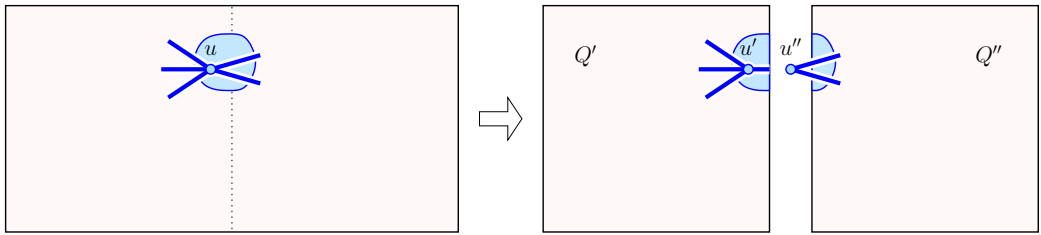}
  \caption{\footnotesize Splitting a rooted lunar EMST.
  For each vertex split into two, we add an anchor whose cost is at most the larger cost of the two new nodes.}
  \label{fig:splitting}
\end{figure}

The cost of the rooted lunar EMST satisfies Conditions~I, II, III,
and by Lemma~\ref{lem:super-additivity}, it also satisfies Condition~VI; that is: it is super-additive with $C_s^{\rm sup} = 0$.
Indeed, we can split the rooted lunar EMST restricted to $[0,1]^2$ into as many rooted trees as we like without increasing the total cost.

\subsection{Mean Closeness of Rooted and Un-rooted Lunar EMSTs}
\label{sec:4.5}

We turn to comparing the costs of the restricted un-rooted and rooted lunar EMSTs discussed in the previous two subsections.
A crucial ingredient in this comparison are the weighted Delaunay mosaics introduced in Section~\ref{sec:3.2}
and, in particular, Lemma~\ref{lem:chromatic_weighted_radii}.
As in the proof of Lemma~\ref{lem:mean_closeness_one}, we need to address the technical assumption in Lemma~\ref{lem:chromatic_weighted_radii}, that the centers of the simplices be in the interval $[0,1]$. 
This assumption can only be violated by a node whose lune enters the unit square through a corner. 
The cost of the anchor connecting this node to the rooted tree is at most a constant smaller than the lowest cost of an edge connecting it to the restricted tree. 
Since there are only four corners, all such nodes are responsible for at most a constant difference in the costs of the two trees.
In what follows, we will ignore these cases and consider only the lunes that enter the unit square through the interiors of the boundary sides.
Clearly,
\begin{align}
  \Costrooted{s}{A_0, A_1, \ldots, A_s; [0,1]^2} &\leq \Costrestr{s}{A_0, A_1, \ldots, A_s; [0,1]^2} ,
\end{align}
because the rooted lunar EMST has strictly more options to get connected than the un-rooted lunar EMST.
Indeed, we permit the rooted tree to have any non-negative number of anchors, so the un-rooted lunar EMST is one of the rooted lunar trees competing for minimizing the cost; see again the footnote of Definition~\ref{dfn:rooted_lunar_EMST}.
The goal is therefore to show that the cost of the un-rooted restricted lunar EMST is not much larger than that of the rooted lunar EMST.
We do this in two steps:  first bounding the expected cost of the anchors, and second bounding the extra cost needed to replace the anchors by other arcs.
The anchors connect the tree to interior points of critical pieces in the overlay of the $s+1$ ($1$-dimensional) weighted Voronoi tessellations.
The expected sum of their weighted radii is the subject of Lemma~\ref{lem:chromatic_weighted_radii}. 
\begin{lemma}
  \label{lem:cost_of_anchors}
  Let $A$ be a stationary Poisson point process with positive intensity in $\Rspace^2$, $s \geq 0$ an integer, and $Q = [0,1]^2$.
  Then the expected total cost---averaged over the point process as well as its colorings---of the anchors in the rooted lunar EMST is at most $O(1)$.
\end{lemma}
\begin{proof}
  Let $\chi \colon A \to [s]$ be a random $(s+1)$-coloring, and consider the overlay of the $s+1$ weighted Voronoi tessellations along the boundary of the unit square.
  The total cost of the anchors is bounded from above by twice the sum of radii of the critical vertices in the dual weighted Delaunay mosaic.
  Since $\chi$ is a random coloring, we can apply Lemma~\ref{lem:chromatic_weighted_radii}, which asserts that this sum is at most $O(1)$.
\end{proof}

To prepare the argument for the second step, call an arc in the un-rooted restricted lunar EMST a \emph{bridge} if it is not also an arc in the rooted lunar EMST.
If there are $k \geq 1$ bridges, removing them splits the tree into $k+1$ subtrees, so we need $k+1$ anchors to reconnect them.
By the optimality of the rooted lunar EMST, the cost of these anchors is at most the cost of the bridges.
But we need an upper bound on the cost of the bridges, which we will get by adding arcs that connect the nodes incident to the anchors in a cycle.
We need a straightforward geometric fact about lunes.
\begin{lemma}
  \label{lem:overlapping_lunes}
  Let $\aaa$ and $\bbb$ be colorful $s$-simplices of $\chi \colon A \to [s]$, $r$ and $q$ larger than the radii of the corresponding smallest enclosing stacks, respectively, and $x \in \Lune{r}{\aaa}$ and $y \in \Lune{q}{\bbb}$.
  Then $\Lune{R}{\aaa} \cap \Lune{R}{\bbb} \neq \emptyset$ for all $R \geq r + q + \frac{1}{2} \Edist{x}{y}$.
\end{lemma}
\begin{proof}
  Observe that $x \in \Lune{r}{\aaa}$ implies that the disk with radius $\ee \geq 0$ centered at $x$ is contained in $\Lune{r+\ee}{\aaa}$.
  The disks with radius $\ee$ centered at $x$ and $y$ touch when $\ee = \frac{1}{2} \Edist{x}{y}$, which implies that $\Lune{r+\ee}{\aaa} \cap \Lune{q+\ee}{\bbb} \neq \emptyset$ if $\ee = \frac{1}{2} \Edist{x}{y}$.
  Since $R \geq r+\ee$ and $R \geq q+\ee$ for this choice of $\ee$, the two lunes with radius $R$ have a non-empty intersection.
\end{proof}

Consider all nodes of the rooted lunar EMST that are connected to $\partial [0,1]^2$ by an anchor.
We add \emph{peripheral} arcs that connect these nodes in a cyclic order following the boundary of the unit square.
To explain the details, let $u$ and $v$ be two such nodes, $x$ and $y$ the endpoints of their anchors on $\partial [0,1]^2$, and $2r$ and $2q$ the costs of these anchors.
We are now in the setting of Lemma~\ref{lem:overlapping_lunes}, which implies that the cost of the connecting peripheral arc is $2R(uv) \leq 2r+2q+\Edist{x}{y}$.
Taking the sum over all pairs of consecutive anchors, we see that the total cost of the peripheral arcs is at most twice the cost of the anchors plus the length of $\partial [0,1]^2$.

\begin{lemma}
  \label{lem:mean_closeness_two}
  Let $A$ be a stationary Poisson point process with positive intensity in $\Rspace^2$, $s \geq 0$ an integer, and $Q = [0,1]^2$.
  The expected absolute difference between the costs of the restricted unrooted and rooted lunar EMSTs is at most some constant; that is:
  $| \Costrestr{s}{A;Q} - \Costrooted{s}{A;Q} | \leq C_s^{\rm two}$.
\end{lemma}
\begin{proof}
  Let $\chi \colon A \to [s]$ be a random $(s+1)$-coloring.
  Take the rooted lunar EMST and add the cycle of peripheral arcs, as explained.
  The expected cost of this graph is at most
  $\Costrooted{}{A_0, A_1, \ldots, A_s; Q}$ plus some constant.
  After removing all anchors and the node that represents $\partial Q$, this graph is still connected.
  By the optimality of the restricted lunar EMST, its cost is at least $\Costrestr{}{A_0, A_1, \ldots, A_s; Q}$.
  Hence, the expected difference between the cost of the restricted lunar EMST and the cost of the rooted lunar EMST is at most some constant.
\end{proof}

In words, Lemma~\ref{lem:mean_closeness_two} establishes Condition~VII for a constant $C_s^{\rm two}$.
We now have Conditions~I to VII, which by \cite[Theorem~5.1]{Yuk98} implies the existence of the asymptotic constant, $c_s$, for a stationary Poisson point process but not yet for points chosen uniformly at random in $[0,1]^2$.

\subsection{The Cost of the Lunar EMST is Combinatorially Smooth}
\label{sec:4.6}

To finally get the existence of the asymptotic constant for the uniform distribution, we need that the cost of the lunar EMST satisfies Condition~VIII.
In a nutshell, this condition requires that adding or deleting a point changes the cost at most by a constant times $1 / \sqrt{n}$ in expectation.
\begin{lemma}
  \label{lem:combinatorial_smoothness}
  Let $s$ be a non-negative integer constant, $A$ a set of $n$ points chosen uniformly at random in $Q = [0,1]^2$, and $A' = A \setminus \{a\}$ for a randomly chosen point $a \in A$.
  Then there is a constant $C_s^{\rm smth}$ such that $\Expect{A, a}{| \Cost{s}{A;Q} - \Cost{s}{A';Q} |} \leq C_s^{\rm smth} / \sqrt{n}$.
\end{lemma}
\begin{proof}
  We first show $\Expect{A}{\Cost{s}{A; Q}} = O(\sqrt{n})$. 
  Let $\chi \colon A \to [s]$ be a random $(s+1)$-coloring.
  By Lemma~\ref{lem:critical_and_non-critical_cells}, an arc of the lunar radius function is critical only if the smallest enclosing stack of the corresponding s + 2 points is empty. 
  Assume $A_0 = \chi^{-1} (0)$ contributes two to the $s+2$ points while every $A_i = \chi^{-1} (i)$ with $1 \leq i \leq s$ contributes only one point. Since the smallest enclosing stack of the arc has to be strictly larger than those of its two nodes, both these points lie on the largest circle in the stack.
  If they are the only two points on the circle, then they are the endpoints of a critical edge in $\Delaunay{A_0}$, and the length of this edge is twice the lunar radius of the critical arc.
  There are at most $n$ points in $A_0$, and since $\chi$ is a random coloring, they are also uniformly distributed in $[0,1]^2$, so Lemma~\ref{lem:radii_and_squared_radii} implies that the expected sum of radii of the critical edges is at most $O(\sqrt{n})$.
  It follows that the expected sum of lunar radii of such arcs is at most $O(\sqrt{n})$.
  If there is a third point on the largest circle in the stack---say of color $1$---then these three points are the vertices of a critical triangle in $\Delaunay{A_0 \cup A_1}$.
  Again by Lemma~\ref{lem:radii_and_squared_radii}, the expected sum of their radii is at most $O(\sqrt{n})$, which implies that the expected sum of lunar radii of the corresponding arcs is at most $O(\sqrt{n})$.
  Multiplying the two bounds by the number of colors and pairs of colors, respectively, we get $O(\sqrt{n})$ as an upper bound for the expected cost.

  \smallskip
  The same argument can be used to prove that the expected sum of lunar radii of the nodes is at most $O(\sqrt{n})$, but we do not need this bound in the first step.

  \smallskip
  The second step shows that removing a point does not change the cost by much, on average.
  To see this, observe that a point affects a node or an arc only if it belongs to the $s+1$ or $s+2$ points that define the node or arc.
  For a point, $a \in A$, write $\Delta \Rnode (a)$ and $\Delta \Rarc (a)$ for the sums of lunar radii of all nodes and arcs that contain $a$, respectively.
  Then
  \begin{align}
    \Expect{A}{\sum\nolimits_{a \in A} \Delta \Rnode (a)} &= O(\sqrt{n}) ; \\
    \Expect{A}{\sum\nolimits_{a \in A} \Delta \Rarc (a)} &= O(\sqrt{n})
  \end{align}
  simply because every node is counted $s+1$ times and every arc $s+2$ times.
  Assuming we pick $a$ randomly among the $n$ points, this implies that the expected change of the cost implied by deleting the nodes and arcs that depend on $a$ is at most some constant divided by $\sqrt{n}$.

  \smallskip
  It remains to show that the same is true for the nodes and arcs that \emph{appear} in the lunar EMST upon removing $a$. 
  Such a node or arc consists of $s+1$ or $s+2$ points of $A \setminus \{a\}$ with $a$ inside their smallest enclosing stack.
  With the same argument as before---but now using the version of Lemma~\ref{lem:radii_and_squared_radii} that allows for one point inside the smallest enclosing circle---we observe that the expected sum of the radii of all such nodes and arcs, over all choices of $a \in A$, is again at most $O(\sqrt{n})$.
  Writing $\Delta \Rnodeprime (a)$ and $\Delta \Rarcprime (a)$ for the sums of lunar radii of all the nodes and arcs that appear upon removing $a$, we have
  \begin{align}
    \Expect{A}{\sum\nolimits_{a \in A} \Delta \Rnodeprime (a)} &= O(\sqrt{n}); \\
    \Expect{A}{\sum\nolimits_{a \in A} \Delta \Rarcprime (a)} &= O(\sqrt{n}) ,
  \end{align}
  because every such node and arc is counted only once.
  Assuming we pick $a$ randomly, this implies that the expected change of the cost caused by adding the nodes and arcs upon removing $a$ is at most some constant divided by $\sqrt{n}$.
  Expectations are additive, which completes the proof of the lemma.
\end{proof}

\section{Discussion}
\label{sec:5}

The main contribution of this paper is the generalization of the Euclidean minimum spanning tree to the family of lunar such trees.
Generalizing a classic probabilistic theorem about the standard EMST, we prove that for each non-negative integer, $s$, there exists a constant, $c_s$, such that the expected cost of the lunar EMST for $n$ points chosen uniformly at random in $[0,1]^2$ and randomly $(s+1)$-colored is $c_s \sqrt{n}$, in the limit when $n$ goes to infinity.
We prove the existence of the constant also for the more general setting in which different colors have different probabilities.
The presented work leaves some open questions.
\smallskip \begin{itemize}
  \item There are three versions of the lunar EMST, and the topological version for $n$ points and $s+1$ colors has at most $O(n)$ vertices; see Lemma~\ref{lem:counting_critical_cells}.
  Is there an algorithm that constructs this tree in time $O(n)$ or perhaps $O(n \log n)$?
  Such an algorithm would have to avoid the explicit construction of the overlay of the $s+1$ Voronoi tessellations, whose size can be quadratic in $n$.
  
  \smallskip
  \item While $c_s$ is known to exist, we do not know its precise value, not even for $c_0$, which is the constant for the standard EMST.
  Lower and upper bounds as well as estimates from computational experiments for $c_1$ are given in \cite{DERS25}; see also Figure~\ref{fig:two-colors}.
  Extend these bounds and estimates to three or more colors.
  Falling short of pinning down the constants exactly, can we prove relations between them?
  For example, is it true that $c_s > c_{s+1}$ for all $s \geq 0$?

  \smallskip
  \item Do asymptotic constants exist for homology degrees beyond $0$?
  Specifically, are there constants $c_{s,p,d}$ such that the expected total degree-$p$ homology of growing $s$-lunes of $n$ points sampled uniformly at random in $[0,1]^d$ and randomly $(s+1)$-colored is $c_{s,p,d} \cdot n^{1 - \sfrac{1}{d}}$ in the limit, when $n$ goes to infinity?
  For $s=0$, $p=1$, $d=2$, the existence of the constant was established in \cite{DERS25}.
\end{itemize} \smallskip

This paper focuses on the linear cost model for the uniform distribution on the unit square. 
Extensions of the results to $3$ and higher dimensions, to quadratic and other cost models, and distributions different from uniform are worthwhile.
We mention that the quadratic cost model corresponds to the $2$-norms of the chromatic persistence diagrams, which show interesting experimental behavior in the plane.




\appendix

\section{Poisson--Delaunay Mosaics in $\Rspace^d$}
\label{app:A}

In this appendix, we first prove a generalized version of Theorem~1 in \cite{ENR17}, and extend it to bounds for the expected sum of radii and squared radii.
We second restate Theorem~1 in \cite{EdNi19}, which generalizes Theorem~1 in \cite{ENR17} to weighted points, and extend it to bounds for the expected sums of weighted radii and squared weighted radii.

\subsection{Generalizing Theorem~1 in \cite{ENR17}}
\label{app:A.1}

We need some concepts and notation to state this theorem and to explain in what sense our result is a generalization.

\smallskip
With probability $1$, the points of a Poisson point process $A$ in $\Rspace^d$ are locally finite and in general position, in which case the Delaunay mosaic is well defined and simplicial.
Calling a sphere \emph{empty} if all sampled points lie on or outside this sphere, the radius function on the Delaunay mosaic
maps each simplex to the radius of the smallest empty sphere that passes through its vertices.
With probability $1$, it is a generalized discrete Morse function \cite{For98,Fre09}, which partitions the Delaunay mosaic into intervals of constant radius.
For non-negative integers $\ell \leq m \leq d$, an \emph{interval of type} `$\ell \leq m$' consists of an $\ell$-face of an $m$-simplex and all simplices between them.
For $\ell = m$, we call the sole simplex in the interval \emph{critical}, while for $\ell < m$, we call the simplices in the interval \emph{non-critical}.
To facilitate counting, we call the center of the smallest empty sphere that passes through the vertices of the $m$-simplex the \emph{center} of the interval.
Given a Borel set, $\Omega \in \Rspace^d$, and a threshold, $r_0 \geq 0$, we write $N_{\ell \leq m \leq d} (r_0)$, $F_{\ell \leq m \leq d} (r_0)$, and $S_{\ell \leq m \leq d} (r_0)$ for the $0$-th, $1$-st, $2$-nd moments (number, sum of radii, sum of squared radii) of the intervals of type $\ell \leq m$ whose centers belong to $\Omega$ and whose radii are at most $r_0$.

\smallskip
To generalize, we relax the requirement that the open ball bounded by a smallest enclosing sphere contains none of the given points.
We do this in two ways.
First, we introduce a fraction, $0 < \eta \leq 1$, and count a point inside a smallest enclosing sphere against selecting the corresponding simplex with probability equal to $\eta$. 
This is equivalent to picking a subset of volume equal to $\eta$ times the volume of the ball bounded by this sphere and requiring that this subset be empty of sampled points.
Second, we introduce an integer, $j \geq 0$, and allow for $j$ points of $A$ in the open ball rather than none.
Accordingly, we embellish the notation and write $N_{\ell \leq m \leq d} (r_0; \eta, j)$ and $N_{\ell \leq m \leq d} (\eta, j)$ if $r_0 = \infty$, etc.
The modification to the proof of Theorem~1 in \cite{ENR17} is minor; see in particular Lemma~6 in \cite{ENR17}, which combines the Slivnyak--Mecke Formula \cite[Chapter 3]{ScWe08} and the Blaschke--Petkantschin Formula \cite[Chapter 7]{ScWe08}.
\begin{theorem}
  \label{thm:counting_generalized}
  Let $A$ be a stationary Poisson point process with intensity $\varrho > 0$ in $\Rspace^d$, $\Omega \subseteq \Rspace^d$ a unit volume Borel set, $0 < \eta \leq 1$, and $j \geq 0$ an integer.
  Then for any $1 \leq \ell \leq m \leq d$ and $r_0 \geq 0$, the expected number of $m$-simplices spanned by $m+1$ points in $A$ whose smallest circumscribed spheres have centers that belong to $\Omega$ and see exactly $m - \ell$ facets of the $m$-simplex, and the $\eta$-fractions of the open balls bounded by these spheres contain $j$ points of $A$ is
  \begin{align}
    \Expect{}{N_{\ell \leq m \leq d} (r_0; \eta, j)}
      &= \frac{1}{j! \cdot \eta^m} \cdot \frac{m!^{d-m} \cdot \varrho}
            {d (m+1) \nu_d^m} \cdot \VVV{\ell \leq m} \norm{\Grass{m}{d}}
         \cdot \gamma (m+j, x) ,
      \label{eqn:A_numbers}
  \end{align}
  in which $x = \eta \varrho \nu_d r_0^d$, $\nu_d$ and $\sigma_d = d \nu_d$ are the $d$- and $(d-1)$-dimensional volumes of the unit ball and unit sphere in $\Rspace^d$,
  $\Grass{m}{d}$ is the Grassmannian of $m$-planes in $\Rspace^d$ with measure $\| \Grass{m}{d} \| = \frac{\sigma_d \cdot \ldots \cdot \sigma_{d-m+1}}{\sigma_1 \cdot \ldots \cdot \sigma_m}$, and
  \begin{align}
    \VVV{\ell \leq m} = \int\limits_{\uuu \in (\Sspace^{m-1})^{m+1}} \Volume{\uuu}^{d-m+1} \One{m-\ell} (\uuu) \diff \uuu
  \end{align}
  is the integrated $(d-m+1)$-st power of the volume of the $m$-simplices spanned by the $(m+1)$-tuples of points in the $(m-1)$-dimensional unit sphere $\Sspace^{m-1}$ such that the origin sees exactly $m-\ell$ of the facets. 
\end{theorem}
\begin{proof}
  We begin by reviewing the spherical Blaschke--Petkantschin Formula, which decomposes an integral over $\Rspace^d$ into an integral over the Grassmannian, $\Grass{m}{d}$, times an integral over an $m$-plane and its orthogonal $(d-m)$-plane:
  \begin{align}
    \int\limits_{\xxx \in {(\Rspace^d)^{m+1}}} \!\!\!\!\!\!\! f(\xxx) \diff \xxx
      &= \int\limits_{L \in \Grass{m}{d}} \int\limits_{h \in L^\bot} \int\limits_{\xxx \in L} f(h+\xxx) \left( m! \Volume{\xxx} \right)^{d-m} \diff \xxx \diff h \diff L 
        \label{eqn:A-BPFormula1} \\
      &\hspace{-1.1in}= \int\limits_{L \in \Grass{m}{d}} \int\limits_{h \in L^\bot} \left[ m! \int\limits_{z \in L} \int\limits_{r \geq 0} \int\limits_{\uuu \in S^{m+1}}
      \!\!\!\!\!\! r^{m^2-1} \Volume{\uuu} f(h+z+r \uuu) 
      \left( m! \Volume{r \uuu} \right)^{d-m} \diff \uuu \diff r \diff z \right] \diff h \diff L
        \label{eqn:A-BPFormula2} \\
      &\hspace{-1.1in}= \int\limits_{L \in \Grass{m}{d}} \int\limits_{z \in \Rspace^d} \int\limits_{r \geq 0} \int\limits_{\uuu \in S^{m+1}} r^{dm-1} f(z+r\uuu)
      \left( m! \Volume{\uuu} \right)^{d-m+1} \diff \uuu \diff r \diff z \diff L ,
        \label{eqn:A-BPFormula3} 
  \end{align}
  in which $S$ is the unit $(m-1)$-sphere in $L$, and we get \eqref{eqn:A-BPFormula2} using Theorem~7.3.1 in \cite[page 287]{ScWe08} to expand the innermost integral in \eqref{eqn:A-BPFormula1}.
  We have $\Volume{r \uuu} = r^m \Volume{\uuu}$, so we get \eqref{eqn:A-BPFormula3} by joining the integration over $L^\bot$ and $L$ and computing the power of the radius as $(m^2-1) + m (d-m) = dm - 1$.

  \smallskip
  We use the spherical Blaschke--Petkantschin Formula to rewrite the Slivnyak--Mecke Formula, which formulates the expectation as an integral.
  For this, let $\nu_d r^d$ be the volume of the ball bounded by the smallest circumscribed sphere of the $m+1$ points in $\xxx \in (\Rspace^d)^{m+1}$, and note that 
  \begin{align}
    \Prob{\eta, j}{\xxx} &= \frac{(\eta \varrho \nu_d r^d)^j}{j!} \cdot e^{- \eta \varrho \nu_d r^d}
  \end{align}
  is the probability a fixed $\eta$-fraction of the open ball bounded by this sphere contains $j$ points of $A$.
  Furthermore, we write $\One{m-\ell}$, $\One{\Omega}$, $\One{r_0}$ for the indicator functions that the center of this sphere sees exactly $m-\ell$ facets of the $m$-simplex, it belongs to $\Omega$, and the radius of this sphere is at most $r_0$:
  \begin{align}
    \Expect{}{N_{\ell \leq m \leq d} (r_0; \eta, j)}
      &= \frac{\varrho^{m+1}}{(m+1)!}
        \int\limits_{\xxx \in (\Rspace^d)^{m+1}}
        \Prob{\eta, j}{\xxx} \cdot \One{m-\ell} \One{\Omega} \One{r_0} \diff \xxx 
        \label{eqn:SMBP1} \\
      &\hspace{-1.3in}= \frac{\varrho^{m+1}}{(m+1)!}
        \int\limits_{L \in \Grass{m}{d}}
        \int\limits_{z \in \Omega}
        \int\limits_{r=0}^{r_0}
        \int\limits_{\uuu \in S^{m+1}} \!\!\!\!\!\! r^{dm-1} \cdot \Prob{\eta, j}{z + r \uuu} \cdot 
        \left( m! \Volume{\uuu} \right)^{d-m+1} \One{m-\ell}
                \diff \uuu \diff r \diff z \diff L
        \label{eqn:SMBP2} \\
      &\hspace{-1.3in}= \frac{{m!}^{d-m} \cdot \varrho^{m+1}}{m+1}
        \cdot \VVV{\ell \leq m} \norm{\Grass{m}{d}}
        \int\limits_{r=0}^{r_0}  r^{dm-1}
        \cdot \frac{(\eta \varrho \nu_d r^d)^j}{j!} \cdot e^{- \eta \varrho \nu_d r^d} \diff r ,
        \label{eqn:SMBP3}
  \end{align}
  in which we get \eqref{eqn:SMBP2} by rewriting two of the indicator functions as limited ranges in the integration.
  To get \eqref{eqn:SMBP3}, we merely note that $\Prob{\eta, j}{z+r\uuu} = \Prob{\eta, j}{r\uuu}$, and that the volume of the relevant ball is $\nu_d r^d$.
  Also, recall that $\norm{\Omega}=1$. 
  To simplify the integral, we set $t =  \eta \varrho \nu_d r^d$, so $\diff t =  \eta \varrho \nu_d d r^{d-1} \diff r$ and $r = \sqrt[d]{t / (\eta \varrho \nu_d)}$, and get
  \begin{align}
    \Expect{}{N_{\ell \leq m \leq d} (r_0; \eta, j)}
      &=  \frac{m!^{d-m} \varrho^{m+1}}{m+1} \cdot \VVV{\ell \leq m} \norm{\Grass{m}{d}}
        \int\limits_{t=0}^{x} 
        \frac{r^{dm-1} (\eta \varrho \nu_d r^d)^j}{j! (\eta \varrho \nu_d d r^{d-1})} \cdot e^{- t} \diff t, 
        \label{eqn:SMBP4} \\
      &= \frac{m!^{d-m} \cdot \varrho^{m+1}}{d (m+1) j!} \cdot \VVV{\ell \leq m} \norm{\Grass{m}{d}} \cdot \frac{1}{(\eta \varrho \nu_d)^m}
      \int\limits_{t=0}^x t^{m+j-1} e^{-t} \diff t 
        \label{eqn:SMBP5} \\
      &= \frac{1}{j! \cdot \eta^m} \cdot \frac{m!^{d-m} \cdot \varrho}{d (m+1) \nu_d^m} \cdot \VVV{\ell \leq m} \norm{\Grass{m}{d}} \cdot \gamma(m+j, x) ,
        \label{eqn:SMBP6}
  \end{align}
  as claimed.
\end{proof}

\subsection{Radii and Squared Radii}
\label{app:A.2}

With comparably little effort, we can use Theorem~\ref{thm:counting_generalized} to derive bounds on the expectations of the sums of radii and squared radii.
\begin{corollary}
  \label{cor:measuring_generalized}
  Let $A$ be a stationary Poisson point process with intensity $\varrho > 0$ in $\Rspace^d$, $\Omega \subseteq \Rspace^d$ a unit volume Borel set, $0 < \eta \leq 1$, and $j \geq 0$ an integer.
  For any $1 \leq \ell \leq m \leq d$ and $r_0 \geq 0$, the expected sums of radii and squared radii of the $m$-simplices spanned by $m+1$ points in $A$ whose smallest circumscribed spheres have centers that belong to $\Omega$ and see exactly $m - \ell$ facets of the $m$-simplex, and the $\eta$-fractions of the open balls bounded by these spheres contain $j$ points of $A$, satisfy
  \begin{align}
    \Expect{}{F_{\ell \leq m \leq d} (r_0; \eta, j)}
      &= \frac{1}{j! \cdot \eta^{m+\frac{1}{d}}} \cdot \frac{m!^{d-m} \cdot \varrho^{1 -  \frac{1}{d}}}{d (m+1) \nu_d^{m+\frac{1}{d}}} \cdot \VVV{\ell \leq m} \norm{\Grass{m}{d}} \cdot \gamma (m+j+\tfrac{1}{d}, x) ; 
        \label{eqn:A_sum_of_radii} \\
    \Expect{}{S_{\ell \leq m \leq d} (r_0; \eta, j)}
      &= \frac{1}{j! \cdot \eta^{m+\frac{2}{d}}} \cdot \frac{m!^{d-m} \cdot \varrho^{1 - \frac{2}{d}}}{d (m+1) \nu_d^{m+\frac{2}{d}}} \cdot \VVV{\ell \leq m} \norm{\Grass{m}{d}} \cdot \gamma (m+j+\tfrac{2}{d}, x) ,
        \label{eqn:A_sum_of_squared_radii}
  \end{align}
  in which $x = \eta \varrho \nu_d r_0^d$.
\end{corollary}
\begin{proof}
  We begin by rewriting \eqref{eqn:A_numbers} in a form that makes the incomplete gamma function explicit:
  \begin{align}
    \Expect{}{N_{\ell \leq m \leq d} (r_0; \eta, j)}
      &= C \int\limits_{t=0}^x t^{m+j-1} e^{-t} \diff t ,
      \label{eqn:A_numbers_rewritten}
  \end{align}
  in which $C$ is the constant factor stated in full in \eqref{eqn:A_numbers}.
  Next we repeat the argument for deriving this relation with one and two extra powers of the radius in the incomplete gamma function.
  Specifically, we set $r(t) = \sqrt[d]{t / (\eta \varrho \nu_d )}$ and get
  \begin{align}
    \Expect{}{F_{\ell \leq m \leq d} (r_0; \eta, j)}
      &= C \int\limits_{t=0}^x r(t) t^{m+j-1} e^{-t} \diff t
       = \frac{C}{\sqrt[d]{\eta \varrho \nu_d}} 
        \int\limits_{t=0}^x t^{m+j-1+\frac{1}{d}} e^{-t} \diff t ; 
      \label{eqn:A_radii} \\
    \Expect{}{S_{\ell \leq m \leq d} (r_0; \eta, j)}
      &= C \int\limits_{t=0}^x r^2(t) t^{m+j-1} e^{-t} \diff t
       = \frac{C}{\sqrt[d/2]{\eta \varrho \nu_d}} 
         \int\limits_{t=0}^x t^{m+j-1+\frac{2}{d}} e^{-t} \diff t .
      \label{eqn:A_squared_radii}
  \end{align}
  Substituting the constant factor in \eqref{eqn:A_numbers} for $C$ and writing the integrals as incomplete gamma functions gives the claimed relations.
\end{proof}

\subsection{Implications for Planar Case}
\label{app:A.3}

We restate Corollary~\ref{cor:measuring_generalized} in elementary terms for the three cases that are relevant to the proof of the Main Theorem in this paper; see in particular Lemmas~\ref{lem:radii_and_squared_radii} and \ref{lem:crossing_the_boundary}.
In each case, we need upper bounds on the sum of radii and sum of squared radii of the critical edges and critical triangles of a stationary Poisson point process $A$ with intensity $n$, in which
\medskip \begin{enumerate}[1.]
  \item $A$ is on $\Rspace^2$ and the smallest enclosing circles are empty and have their centers in $[0,1]^2$;
  \item $A$ is on $[0,1]^2$ and the smallest enclosing circles are empty;
  \item $A$ is on $[0,1]^2$ and the smallest enclosing circles enclose at most one point.
\end{enumerate} \medskip
We thus have $d=2$, $\varrho=n$, $r_0=\infty$ in all three cases, $\eta=1$, $j=0$ in Case~1, $\eta=\frac{1}{2}$, $j=0$ in Case~2, and $\eta=\frac{1}{2}$, $j=1$ in Case~3.
Indeed, if $A \subseteq [0,1]^2$, then at least half the disk bounded by a smallest enclosing circle belongs to $[0,1]^2$.
The bounds on the expectations are given in the first four columns of Table~\ref{tbl:bounds}.
The last two columns give upper bounds on the expectations of the sum of radii under the additional constraint that the smallest enclosing circle is not fully inside $[0,1]^2$.
{\renewcommand{\arraystretch}{1.4}%
\begin{table}[hbt]
  \vspace{0.1in}
  \centering
  \begin{tabular}{c||cccc|cc}
    $\mathbb E$ & $F_{1\leq1\leq2}$ & $F_{2\leq2\leq2}$ & $S_{1\leq1\leq2}$ & $S_{2\leq2\leq2}$ & $F_{1\leq1\leq2}^\partial$ & $F_{2\leq2\leq2}^\partial$ \\ \hline \hline
    $\eta=1, j=0$ & $\sqrt{n}$ & $\frac{3}{4} \sqrt{n}$ & $\sfrac{2}{\pi}$ & $\sfrac{2}{\pi}$ & $\sfrac{16}{\pi}$ & $\sfrac{16}{\pi}$ \\
    $\eta=\frac{1}{2}, j=0$ & $\sqrt{8n}$ & $\sqrt{18n}$ & $\sfrac{8}{\pi}$ & $\sfrac{16}{\pi}$ & $\sfrac{64}{\pi}$ & $\sfrac{128}{\pi}$ \\
    $\eta=\frac{1}{2}, j=1$ & $\sqrt{18n}$ & $\sqrt{\frac{450}{4} n}$ & $\sfrac{16}{\pi}$ & $\sfrac{48}{\pi}$ & $\sfrac{128}{\pi}$ & $\sfrac{384}{\pi}$
  \end{tabular}
  \caption{The upper bounds on the expected sums of radii and squared radii of the critical edges and triangles give the upper bounds for the Cases~1, 2, 3.
  The last two columns give upper bounds for the sums of radii when we consider only those whose smallest enclosing circles are not fully inside $[0,1]^2$.
  }
  \label{tbl:bounds}
\end{table}}
To see the bounds in the first four columns, it suffices to recall a few values of the Gamma function, $\Gamma(\frac{3}{2}) = \frac{1}{2} \sqrt{\pi}$, $\Gamma(2) = 1$, $\Gamma(\frac{5}{2}) = \frac{3}{4} \sqrt{\pi}$, $\Gamma(3) = 2$, $\Gamma(\frac{7}{2}) = \frac{15}{8} \sqrt{\pi}$, $\Gamma(4) = 6$, and plug the appropriate constants into Corollary~\ref{cor:measuring_generalized}, which are $\nu_2 = \pi$, $\Grass{1}{2} = \pi$, $\Grass{2}{2} = 1$, $\VVV{1 \leq 1} = 8$, and $\VVV{2 \leq 2} = 6 \pi^2$.
In contrast, the right two columns of the table need a formal proof.
\begin{lemma}
  \label{lem:crossing_the_boundary_details}
  Let $A$ be a stationary Poisson point process with positive intensity in $\Rspace^2$.
  For Cases~1, 2, 3, the expected sum of the radii of the smallest enclosing circles of the critical edges and triangles with center in $[0,1]^2$ that are not fully inside $[0,1]^2$ are bound from above by constants independent of the density, as listed in the last two columns of Table~\ref{tbl:bounds}.
\end{lemma}
\begin{proof}
  Apart from minor differences, the argument is the same in all three cases, so we limit ourselves to Case~1, in which $\eta=1$ and $j=0$.
  Letting $x$ and $y$ be the endpoints of an edge, we write $\Enc{x}{y}$ for their smallest enclosing circle, and $\midpoint{x}{y} = \frac{1}{2} (x+y)$ for its center.
  We consider points $a, b \in A$ such that the following three conditions are all satisfied:
  \smallskip \begin{enumerate}[(i)]
    \item $\Enc{a}{b}$ does not enclose any point of $A$;
    \item $\midpoint{a}{b} \in [0,1]^2$; 
    \item $\Enc{a}{b} \not\subseteq [0,1]^2$.
  \end{enumerate} \smallskip
  We are interested in $F_{1 \leq 1 \leq 2}^\partial$, which is the sum of radii of all such edges, and compare this with $S_{1 \leq 1 \leq 2}$, which we define as the sum of squared radii of all edges that satisfy (i) and (ii) but not necessarily (iii).
  Particularly, we claim that
  \begin{align}
    \Expect{}{F_{1 \leq 1 \leq 2}^\partial} &\leq 8 \cdot \Expect{}{S_{1 \leq 1 \leq 2}},
      \label{eqn:FS}
  \end{align}
  which implies the second to the last entry in the first row of Table~\ref{tbl:bounds}.
  To see \eqref{eqn:FS}, write $F(t)$ for the sum of radii of all edges connecting points $a, b \in A$ for which Condition (i) and (ii) is satisfied, and (iii) is changed to $\Enc{a}{b}$ having a non-empty intersection with the vertical line of points $x_1 = t$.
  We claim
  \begin{equation}
    \Expect{}{F(0)} \leq
    \Expect{}{\int_{t=0}^1 F(t) \diff t} \leq
    2 \Expect{}{S_{1 \leq 1 \leq 2}}.
  \end{equation}
  The second inequality holds for any point set, since every pair $a,b\in A$ with radius $r$ satisfying Conditions~(i) and (ii) contributes at most $2r\cdot r$ to the left hand side, and exactly $2r^2$ to the right hand side. 
  For the first inequality, fix $t\in[0, 1]$ and observe that $\Expect{}{F(0)} \leq \Expect{}{F(t)}$. 
  Indeed, the contributions of the edges with $z(a,b)$ between the vertical lines $x_1 = 0$ and $x_1 = t$ are symmetrical to both expectations by the symmetry of Poisson point process, while the edges with $z(a,b)$ to the right of $x_1 = t$ contribute to the latter whenever they contribute to the former.

  \smallskip
  We get the claimed inequality \eqref{eqn:FS} because there are four sides through which the enclosing disks can exit the unit square.
  Finally, we use the same argument to prove $\Expect{}{F_{2 \leq 2 \leq 2}^\partial} \leq 8 \cdot \Expect{}{S_{2 \leq 2 \leq 2}}$ for the critical triangles, which implies the last entry in the first row of Table~\ref{tbl:bounds}.
\end{proof}

\subsection{Restating Theorem~1 in \cite{EdNi19}}
\label{app:A.4}

Voronoi tessellations and Delaunay mosaics can also be defined for points with real weights.
This is best described for a $k$-plane in $d$-dimensional space, denoted $\Rspace^k \hookrightarrow \Rspace^d$, in which we assume $0 \leq k < d$.
Given a locally finite set, $A \subseteq \Rspace^d$, we call the intersection of the Voronoi tessellation of $A$ with $\Rspace^k$ a \emph{weighted Voronoi tessellation}, and its dual a \emph{weighted Delaunay mosaic}.
Both can also be defined intrinsically:
think of each point $a \in A$ as a \emph{weighted point}, denoted $a'$, whose \emph{location} is the projection of $a$ onto $\Rspace^k$, and whose \emph{weight} is the negative squared distance of $a$ from $\Rspace^k$.
We write $A'$ for the collection of weighted points, and $\Delaunay{A'}$ for its weighted Delaunay mosaic in $\Rspace^k$.
In the interest of brevity, we omit details and refer to \cite[Chapter~III]{EdHa10} for further background.

\smallskip
Generically, $\Delaunay{A'}$ is a simplicial complex that is geometrically realized in $\Rspace^k$.
Important for us is the \emph{weighted radius function}, $g \colon \Delaunay{A'} \to \Rspace$, which maps each simplex to the minimum radius for which the spheres centered at the points in $A$ reach a point of the dual cell in the weighted Voronoi tessellation.
Since all points have a non-negative distance from $\Rspace^k$, the weighted radii assigned by $g$ are necessarily non-negative.
Like its unweighted counterpart, the weighted radius function is generalized discrete Morse, so its level sets partition $\Delaunay{A'}$ into intervals, and we distinguish between \emph{critical} and \emph{non-critical} simplices, as before.
To facilitate counting, we define the \emph{center} of an interval as the first point of the dual cell in the weighted Voronoi tessellation reached by the spheres centered at the points in $A$.
Equivalently, it is the center of the smallest empty sphere centered in $\Rspace^k$ that passes through the corresponding points in $A$.
Extending the notation by one parameter, we write $N_{\ell \leq m \leq k \leq d} (r_0)$ for the number of intervals of type $\ell \leq m$ whose weighted radii are at most $r_0$ and whose centers belong to a given Borel set in $\Rspace^k$.
We are now ready to restate Theorem~1 in \cite{EdNi19}.
\begin{theorem}
  \label{thm:weighted_counting}
  Let $A$ be a stationary Poisson point process with intensity $\varrho > 0$ in $\Rspace^d$, $\Rspace^k \hookrightarrow \Rspace^d$ a $k$-plane with $k < d$, $\Delaunay{A'}$ the corresponding weighted Delaunay mosaic in $\Rspace^k$, and $\Omega \subseteq \Rspace^k$ a unit volume Borel set.
  Then for any $0 \leq \ell \leq m \leq k$ and $r_0 \geq 0$, the expected number of intervals of type $\ell \leq m$ with weighted radius at most $r_0$ and center in $\Omega$ is
  \begin{align}
    \Expect{}{N_{\ell \leq m \leq k \leq d} (r_0)}
      &= \frac{C_{\ell \leq m \leq k \leq d} \cdot \varrho^{\frac{k}{d}}}{\Gamma (m+1-\frac{k}{d})}
      \cdot \gamma (m+1-\tfrac{k}{d}, x),
      \label{eqn:A3}
  \end{align}
  in which $x = \varrho \nu_d r_0^d$ and $C_{\ell \leq m \leq k \leq d}$ is a constant.
\end{theorem}
We refer to \cite{EdNi19} for details about the constants.
For $x = \infty$, the complete and incomplete gamma functions in \eqref{eqn:A3} cancel each other.
In this paper, we need them only in $\Rspace^1 \hookrightarrow \Rspace^2$, for which $C_{0 \leq 0 \leq 1 \leq 2} = 1$, $C_{0 \leq 1 \leq 1 \leq 2} = \frac{4}{\pi} - 1$, and $C_{1 \leq 1 \leq 1 \leq 2} = 1$.
The first two constants imply that for a stationary Poisson point process with intensity $\varrho > 0$ in $\Rspace^2$, the expected density of non-redundant weighted points in $\Rspace^1$ is only $\frac{4}{\pi} \sqrt{\varrho}$.
Indeed, most points sampled in $\Rspace^2$ will have domains in the Voronoi tessellation that do not intersect $\Rspace^1$ and are therefore redundant.
We remark that the locations of the non-redundant weighted points in $\Rspace^1$ are not necessarily sampled according to a Poisson point process in $\Rspace^1$.

\subsection{Weighted Radii and Squared Weighted Radii}
\label{app:A.5}

Like in the unweighted case, it is comparably easy to derive the expected sums of weighted radii and squared weighted radii from Theorem~\ref{thm:weighted_counting}.
Following the earlier convention, we write $F_{\ell \leq m \leq k \leq d} (r_0)$ and $S_{\ell \leq m \leq k \leq d} (r_0)$ for these sums.
\begin{corollary}
  \label{cor:weighted_measuring}
  Let $A$ be a stationary Poisson point process with intensity $\varrho > 0$ in $\Rspace^d$, $\Rspace^k \hookrightarrow \Rspace^d$ a $k$-plane with $k < d$, $\Delaunay{A'}$ the corresponding weighted Delaunay mosaic in $\Rspace^k$, and $\Omega \subseteq \Rspace^k$ a unit volume Borel set.
  Then for any $0 \leq \ell \leq m \leq k$ and $r_0 \geq 0$, the expected sums of weighted radii and squared weighted radii of the intervals of type $\ell \leq m$ with weighted radius at most $r_0$ and center in $\Omega$ is
  \begin{align}
    \Expect{}{F_{\ell \leq m \leq k \leq d} (r_0)}
      &= \frac{C_{\ell \leq m \leq k \leq d} \cdot \varrho^{\frac{k-1}{d}}}{\sqrt[d]{\nu_d} \cdot \Gamma (m+1-\frac{k}{d})}
      \cdot \gamma (m+1-\tfrac{k-1}{d}, x) ; 
        \label{eqn:A4_sum_of_weighted_radii} \\
    \Expect{}{S_{\ell \leq m \leq k \leq d} (r_0)}
      &= \frac{C_{\ell \leq m \leq k \leq d} \cdot \varrho^{\frac{k-2}{d}}}{\sqrt[d/2]{\nu_d} \cdot \Gamma (m+1-\frac{k}{d})}
      \cdot \gamma (m+1-\tfrac{k-2}{d}, x) ,
        \label{eqn:A4_sum_of_squared_weighted_radii}
  \end{align}
  in which $x = \varrho \nu_d r_0^d$ and $C_{\ell \leq m \leq k \leq d}$ is the same constant as in Theorem~\ref{thm:weighted_counting}.
\end{corollary}
\begin{proof}
  To simplify the notation, we set $C = C_{\ell \leq m \leq k \leq d} / {\Gamma (m+1-\frac{k}{d})}$, and rewrite the relation \eqref{eqn:A3} in a form that makes the incomplete gamma function explicit:
  \begin{align}
    \Expect{}{N_{\ell \leq m \leq k \leq d} (r_0)}
      &= C \varrho^{\frac{k}{d}}
      \int_{t=0}^x t^{m-\frac{k}{d}} e^{-t} \diff t .
      \label{eqn:A3restated}
  \end{align}
  Next we repeat the argument for deriving this relation with one and two extra powers of the radius in the incomplete gamma function.
  Specifically, we set $r(t) = \sqrt[d]{t / (\varrho \nu_d)}$ and get
  \begin{align}
    \Expect{}{F_{\ell \leq m \leq k \leq d} (r_0)} 
      &= C \varrho^{\frac{k}{d}} \int\nolimits_{t=0}^x r(t) t^{m-\frac{k}{d}} e^{-t} \diff t 
      = \frac{C \varrho^{\frac{k}{d}}}{\sqrt[d]{\varrho \nu_d}} \int\nolimits_{t=0}^x t^{m-\frac{k-1}{d}} e^{-t} \diff t ; \\
    \Expect{}{S_{\ell \leq m \leq k \leq d} (r_0)} 
      &= C \varrho^{\frac{k}{d}} \int\nolimits_{t=0}^x r^2(t) t^{m-\frac{k}{d}} e^{-t} \diff t 
      = \frac{C \varrho^{\frac{k}{d}}}{\sqrt[d/2]{\varrho \nu_d}} \int\nolimits_{t=0}^x t^{m-\frac{k-2}{d}} e^{-t} \diff t .
  \end{align}
  Undoing the substitution with the constant, $C$, writing the integrals as incomplete gamma functions, and simplifying the factors in front of them, we get the two claimed relations.
\end{proof}

As an example consider a stationary Poisson point process with positive intensity in $\Rspace^2$, and a line $\Rspace^1 \hookrightarrow \Rspace^2$, which intersects the Voronoi tessellation of the process in a ($1$-dimensional) weighted Voronoi tessellation.
So $d=2$ and $k=1$.
We get the expected sums of weighted radii of the critical vertices, non-critical vertex-edge pairs, and critical edges with centers in a unit-length interval $\Omega \subseteq \Rspace^1$ by plugging the relevant parameters into \eqref{eqn:A4_sum_of_weighted_radii}.
Specifically, we have $\nu_2 = \pi$, $C_{0 \leq 0 \leq 1 \leq 2} = 1$, $C_{0 \leq 1 \leq 1 \leq 2} = \frac{4}{\pi} - 1$, and $C_{1 \leq 1 \leq 1 \leq 2} = 1$, which gives the three constant expectations listed in Table~\ref{tbl:weighted_bounds}.
{\renewcommand{\arraystretch}{1.4}
\begin{table}[hbt]
  \vspace{0.0in}
  \centering
  \begin{tabular}{c||ccc}
    $\mathbb E$ & $F_{0 \leq 0 \leq 1 \leq 2}$ & $F_{0 \leq 1 \leq 1 \leq 2}$ & $F_{1 \leq 1 \leq 1 \leq 2}$ \\ \hline \hline
    $d = 2, k = 1$ & $\sfrac{1}{\pi}$ & $\frac{2}{\pi} (\frac{4}{\pi} - 1)$ & $\sfrac{2}{\pi}$ 
  \end{tabular}
  \caption{\emph{From left to right}: the expected sum of weighted radii of the critical vertices, non-critical vertex-edge pairs, and critical edges with centers in $\Omega$.}
  \label{tbl:weighted_bounds}
\end{table}}

\section{Colorful Simplices of the Chromatic Delaunay Complex} \label{app:B}

In this appendix, we prove Proposition~\ref{prop:chromatic_vs_lunar} and argue that the filtration of the lunar Delaunay mosaic is a generalized discrete Morse function, assuming general position of the points. 
The proofs are for points in $d$-dimensional Euclidean space rather than just a plane.

\smallskip
Fix $s+1$ disjoint finite sets, $A_0, A_1, \dots, A_s \subseteq \Rspace^d$, let $A = \bigcup_{i=0}^s A_i$ be the union and $\chi \colon A \to [s]$ the map that sends points in $A_i$ to $i$ for each $0 \leq i \leq s$.
Let us recall the constructions described in Section~\ref{sec:2}, in preparation for the proofs but also to ascertain that they all generalize to $\Rspace^d$.
We can still consider the overlay of Voronoi tessellations, $\Voronoi{A_0, A_1, \dots, A_s}$, with cells $\nu = \nu_0 \cap \nu_1 \cap \dots \cap \nu_s$, where $\nu_i$ is a Voronoi cell of $\Voronoi{A_i}$ for each $0\leq i \leq s$.
It decomposes $\Rspace^d$ into domains such that for all points in the interior of a domain $f_{\max}(x) = \max_{0 \leq j \leq s} \min_{a\in A_j} \norm{x - a}$ is defined by the same $s+1$ points.
Similarly, the dual lunar Delaunay mosaic is well-defined, and filtered by $\Lambda \colon \Delaunay{A_0, A_1, \ldots, A_s} \to \Rspace$, which maps each cell to the minimum value of $f_{\max}$ over the points in the dual cell of the overlay of Voronoi tessellations.
In other words, if $r$ is the radius at which the first lune touches a cell $\nu \in \Voronoi{A_0, A_1, \dots, A_s}$, its dual, $\nu^\ast$, appears in $\Delaunay[r]{A_0, A_1,\dots,A_s} = \Lambda^{-1} [0,r]$.
Using the Nerve Theorem and collapses of simplices to corresponding cells in the lunar Delaunay mosaic, it is not difficult to see that 
\begin{align}
  \Delaunay[r]{A_0,A_1,\dots,A_s} &\simeq \bigcup\nolimits_{\aaa \in A_0 \times \dots \times A_s} \Lune{r}{\aaa} ,
\end{align}
like in Lemma~\ref{lem:lunar_alpha_vs_lunar_space}. 
The construction of the chromatic Delaunay mosaic, $\Delaunay{\chi}$, is similar. 
It contains a simplex $\sigma \subseteq A$ whenever the domains, $\domain{a, A_{\chi(a)}}$, of the $a \in \sigma$ have a non-empty common intersection.
Furthermore, the filtration value, $\alpha(\sigma)$, is the minimum radius, $r$, at which the balls of radius $r$ centered at the points $a \in \sigma$ have a non-empty common intersection within the intersection of the Voronoi domains.
Finally, a simplex $\sigma \subseteq A$ is \emph{colorful} if $\chi(\sigma) = [s]$, we write $\Delaunaycol{\chi}$ for the collection of colorful simplices in $\Delaunay{\chi}$, and observe that for a colorful simplex, the intersection of $\domain{a; A_{\chi(a)}}$, for $a \in \sigma$, is a cell in the overlay of Voronoi tessellations, $\Voronoi{A_0, A_1, \dots, A_s}$.
  
\smallskip
For the following proposition we assume that the points in $A$ are in generic position.
More precisely, we need what we call \emph{chromatic genericity}, as discussed in more detail below.
\begin{proposition}
  \label{prop:colorful_chromatic_vs_lunar}
  Let $A \subseteq \Rspace^d$ be a finite set of points in general position, and $\chi \colon A \to [s]$ an $(s+1)$-coloring.
  Then there is a bijection $\varphi \colon \Delaunaycol{\chi} \to \Delaunay{A_0, A_1, \dots, A_s}$ between the colorful simplices of $\Delaunay{\chi}$ and the cells in the lunar Delaunay mosaic, such that for each $\sigma \in \Delaunaycol{\chi}$:
  \begin{itemize}
    \item $\dim \sigma = \dim \varphi(\sigma) + s$;
    \item $\alpha(\sigma) = \Lambda(\varphi(\sigma))$;
    \item $\varphi(\partial\sigma) = \partial\varphi(\sigma)$,
  \end{itemize}
  for each $\sigma \in \Delaunaycol{\chi}$, in which $\alpha$ and $\Lambda$ are the radius functions of the two mosaics.
\end{proposition}
\begin{proof}
    Let $\sigma\subseteq A$ be a colorful simplex, $\nu = \bigcap_{a \in \sigma} \domain{a; A_{\chi(a)}}$, and $V$ be the set of colorful $s$-simplices
    contained in $\sigma$. 
    Observe that $\nu$ is also the intersection of the domains of the overlay of Voronoi tessellations; that is: $\nu = \bigcap_{\vvv \in V} \domain{\vvv}$. 
    Then $\sigma \in \Delaunaycol{\chi}$ iff $\nu\neq\emptyset$, which is equivalent to $\nu$ being a cell of the overlay. 
    Writing $\nu^\ast$ for the dual of $\nu$, we set $\varphi(\sigma) = \nu^\ast$ for every $\sigma \in \Delaunaycol{\chi}$. 
    This defines the desired bijection from $\Delaunaycol{\chi}$ to $\Delaunay{A_0, \dots, A_s}$. 
    We check the three claimed properties.
    \smallskip \begin{itemize}
        \item Since the points are in general position, the codimension of $\nu$ is $\dim\sigma - s$; see \cite[Lemma~4.2(a)]{CDES26}, which is also the dimension of $\nu^\ast = \varphi(\sigma)$.
        Hence, $\dim\sigma = \dim \varphi(\sigma) + s$, as claimed.
        
        \item For any point $x \in \nu$ and color $0 \leq i \leq s$, $x$ is equidistant to all points in $\sigma \cap A_i$. 
        Therefore, the minimum radius, $r$, at which one of $\Lune{r}{\vvv}$, for $\vvv \in V$, touches $\nu$, is also the minimum radius at which all disks centered at points in $\sigma$ touch the same point in $\nu$. 
        Both events happen at $r = \min_{x\in\nu}f_{\max}(x)$.
        The former event defines $\Lambda$, and the latter defines $\alpha$, so $\Lambda(\nu^\ast) = \alpha(\sigma)$.
        
        \item Each colorful facet of $\sigma$ is a simplex $\tau = \sigma \setminus \{a\}$, in which $a \in \sigma$ is not the only point of its color in $\sigma$.
        Hence, the codimension of the corresponding cell in $\Voronoi{A_0, A_1, \ldots, A_s}$ is one less than that of $\nu$.
        It is therefore a cofacet of $\nu$, and its dual is a facet of $\nu^\ast$ in $\Delaunay{A_0, A_1, \dots, A_s}$.
    \end{itemize} \smallskip
    The last argument can be reversed, which implies that the colorful facets of $\sigma \in \Delaunaycol{\chi}$ are exactly the facets of $\nu^\ast$ in the lunar Delaunay mosaic. 
\end{proof}

We next argue that $\Lambda \colon \Delaunay{A_0, A_1, \dots, A_s}\to\Rspace$ is a generalized discrete Morse function, as defined in \cite{For98,Fre09}, provided the points satisfy a sufficiently strong genericity condition. 
An \emph{interval} of $\Lambda$ is a collection of cells $\sigma \in \Delaunay{A_0, A_1, \dots, A_s}$ that satisfy $\sigma_{\min} \subseteq \sigma \subseteq \sigma_{\max}$ and $\Lambda(\sigma_{\min}) = \Lambda(\sigma) = \Lambda(\sigma_{\max})$, for some $\sigma_{\min}, \sigma_{\max}$.
The function $\Lambda$ is \emph{generalized discrete Morse} if the mosaic partitions into intervals that are maximal with respect to inclusion.
A cell is called \emph{critical} if it forms a (singular) interval by itself, and \emph{non-critical}, otherwise.
An important consequence of being generalized discrete Morse is that the homotopy type of $\Lambda^{-1} [0,r]$ changes iff a critical cell is added.

\smallskip
By \cite[Thm~4.6]{CDES26} (see also \cite{NCBJ24}), the radius function $\alpha$ on $\Delaunay{\chi}$ is a generalized discrete Morse function, provided the points are chromatically generic.
To show that $\Lambda$ is a generalized discrete Morse function, we use the proposition above and then show that the intervals of $\alpha$ restricted to the colorful simplices are still intervals. 
We first introduce a sufficiently strong notion of genericity, and then prove the claim.
\begin{definition}[Strong chromatic genericity]
  \label{dfn:strong_chromatic_genericity}
    Let $A \subseteq \Rspace^d$ be finite.
    For any $s+1$ disjoint subsets $X_0, X_1, \ldots, X_s \subseteq A$, let $E_i$ be the intersection of the bisectors of the points in $X_i$, for $0 \leq i \leq s$, and $E = E_0 \cap E_1 \cap \ldots \cap E_s$.
    Furthermore, let $x_i \in X_i$ be arbitrary choice of points, for $0 \leq i \leq k \leq \dim E$, and let $x_i'$ be the orthogonal projections of $x_i$ to $E$.
    We say that $A\subseteq \Rspace^d$ is \emph{strongly chromatically generic} if
    \begin{enumerate}[(i)]
      \item[(i)] any $k \leq d+1$ points in $A$ are affinely independent;
        \label{dfn:strong_chromatic_genericity-i}
      \item[(ii)] any $k+1$ concentric $(d-1)$-spheres intersected with an affine $p$-plane contain at most $p+k+1$ points from $A$;
        \label{dfn:strong_chromatic_genericity-ii}
      \item[(iii)] for any choice of the $s+1$ disjoint sets and a point in each, $x_0', x_1', \dots, x_k'$ are affinely independent;
        \label{dfn:strong_chromatic_genericity-iii}
      \item[(iv)] for any choice of $s+2$ disjoint sets, $X_0, X_1, \ldots, X_s, \{y_0, y_1\} \subseteq A$, the center of the smallest circumcircle of $x_0, x_1, \dots, x_k$ centered at $E$ (if it exists) does not lie on the bisector of $y_0$ and $y_1$.
        \label{dfn:strong_chromatic_genericity-iv}
    \end{enumerate}
\end{definition}
We note that Condition~(ii) is the \emph{chromatic genericity} as defined in \cite[Def.~4.1]{CDES26}. 
See \cite[Lem.~4.2]{CDES26} for useful characterizations.
Before using Definition~\ref{dfn:strong_chromatic_genericity}, let us make sure that it deserves to be called a genericity condition.
\begin{lemma}
  \label{lem:measure-zero_subset}
  The space of concatenated coordinates of $n$ points in $\Rspace^d$ that do not satisfy the four conditions in Definition~\ref{dfn:strong_chromatic_genericity} is a zero measure subset of $\Rspace^{nd}$.
\end{lemma}
\begin{proof}
  Consider the coordinate space $\Rspace^{nd}$ of $n$ points in $d$-dimensional space. Clearly, configurations violating Conditiion~(i) are of measure zero, and according to \cite[Lem.~4.3]{CDES26} so are configurations violating Condition~(ii).
  We argue that the configurations that violate (iii) or (iv)
  can be described as the zero set of a system of rational functions.

  \smallskip
  Indeed, the center, $c$, of the smallest circumsphere of the points $x_0, x_1, \dots, x_k$ centered at $E$, if it exists, can be expressed as a rational function in the coordinates of those points. 
  Set vectors $u_i = x_i' - x_0'$ and weights $w_i = \norm{x_i' - x_0'}^2$ for $0 \leq i \leq k$. 
  The center $c = \sum_{i=1}^k \gamma_i u_i$ satisfies $\norm{c}^2 + w_0 = \norm{u_i - c}^2 + w_i$ for $1 \leq i \leq k$. 
  Expanding the squared norm and canceling $\norm{c}^2$, this system of equations can be written as $2 U^T U \gamma = b$, where $U$ is the matrix with $u_1, u_2, \dots, u_k$ as columns, $\gamma = (\gamma_1, \gamma_2, \dots, \gamma_k)$ a vector of variables, and $b$ the vector of constants, $b_i = \norm{u_i}^2 + w_i - w_0$, for $1 \leq i \leq k$. 
  The Gram matrix $U^T U$ is invertible iff the orthogonal projections $x_0', x_1', \dots, x_k'$ are affinely independent. 
  In that case, we can write $c = c_0 + \frac{1}{2} U (U^T U)^{-1} b$.

  \smallskip
  If Conditions (iii) or (iv) are violated, then either (iii) is violated---which is expressed by the zero set of a rational function $\det(U^T U)$---or (iv) is violated---which is expressed by the zero set of a rational function $\norm{y_0 - c}^2 - \norm{y_1 - c}^2$. 
  Overall we have only finitely many choices of the sets $X_i$ and points $x_i$ and $y_0, y_1$, so the space of configurations violating (iii) or (iv) is the zero set of finitely many non-trivial rational functions and as such a set of measure zero.
\end{proof}

The motivation for strong chromatic genericity formulated in Definition~\ref{dfn:strong_chromatic_genericity} is the following lemma, which says that removing a non-singleton color from a colorful simplex strictly decreases the radius.
\begin{lemma}
  \label{lem:remove_nontrivial_color}
  Let $A$ be strongly chromatically generic, and $\sigma \in \Delaunaycol{\chi}$ a maximal colorful simplex with respect to the filtration $\alpha$; that is: $\alpha(\tau) > \alpha(\sigma)$ for every proper coface $\tau \supseteq \sigma$ in $\Delaunaycol{\chi}$.
  Suppose there is a color $0 \leq i \leq s$ such that $\sigma_i = \chi^{-1}(i)\cap \sigma$ contains at least two points. 
  Then $\alpha(\sigma \setminus \sigma_i) < \alpha(\sigma)$.
\end{lemma}
\begin{proof}
  Let $\sigma' = \sigma \setminus \sigma_i$ be the simplex without the points of color $i$.
  Consider the corresponding intersections of Voronoi domains, and the filtration values defined by minimizing over their points:
  \begin{align}
    \nu &= \bigcap\nolimits_{a \in \sigma} \domain{a; A_{\chi(a)}} , \hspace{0.22in} \alpha(\sigma)^2 = \min_{x\in\nu} \max_{a \in \sigma} \Edist{x}{a}^2; \\
    \nu' &= \bigcap\nolimits_{a \in \sigma'} \domain{a , A_{\chi(a)}}, \quad \alpha(\sigma')^2 = \min_{x \in \nu'} \max_{a \in \sigma'} \Edist{x}{a}^2 .
  \end{align}
  Clearly, $\nu \subseteq \nu'$. 
  The filtration values are obtained by minimizing strictly convex functions over convex domains, so the points $c \in \nu$ and $c' \in \nu'$ at which the functions attain their minima are unique.

  \smallskip
  We show that $c \neq c'$.
  The maximality of $\sigma$ implies that for each color $0 \leq j \leq s$ the point $c$ lies in the interior of the intersection of $\domain{a; A_j}$ over $a\in\chi^{-1}(j)\cap\sigma$, as otherwise more points could be added to $\sigma$ without increasing the filtration value.
  Condition~(ii) of Definition~\ref{dfn:strong_chromatic_genericity} implies that those $j$-color Voronoi cells intersect generically, and assuming $\card\sigma_i \geq 2$, this implies that there exists a direction $w$ and a small $\ee > 0$ such that $\nu'$ lies on both sides of $\nu$; that is: $c \pm \ee w \in \nu' \setminus \nu$. 
  We claim that the value of $\max_{a \in \sigma'} \Edist{x}{a}^2$ decreases if we move from $c$ either in the direction of $w$ or $-w$. 
  To see this, it suffices to show that the point $c''$ where the function attains its minimum over the affine hull $E=\affine{\nu'}$ does not lie in $\affine{\nu}$.

  \smallskip
  The point $c''$ is, by definition, the center of the smallest ball centered at $E$ enclosing $\sigma'$. 
  Necessarily, it is then the center of the smallest circumsphere centered at $E$ of a subset of points $\sigma'' \subseteq \sigma'$. 
  Since all points of a single color are equidistant from each point of $E$, we can assume that $\sigma''$ only contains one point of any given color. 
  Applying Condition~(iv) to the affine space $E$, we see that the points $\sigma''$ take on the role of the points $x_0, x_1, \dots, x_k$, and any two points of $\sigma_i$ take on the role of $y_0, y_1$, so $c'' \not\in \affine{\nu}$.
\end{proof}

\begin{theorem}
  \label{thm:generalized_discrete_Morse}
  Let $A$ be a finite set of strongly chromatically generic points in $\Rspace^d$.
  Then the filtration function $\Lambda \colon \Delaunay{A_0, A_1, \dots, A_s} \to \Rspace$ is generalized discrete Morse.
\end{theorem}
\begin{proof}
  By \cite[Thm~4.6]{CDES26}, the filtration function $\alpha \colon \Delaunay{\chi} \to \Rspace$ is generalized discrete Morse. 
  According to Proposition~\ref{prop:colorful_chromatic_vs_lunar}, it suffices to show that the restriction of $\alpha$ to $\Delaunaycol{\chi}$ is generalized discrete Morse.

  \smallskip
  Let $I$ be an interval of $\alpha$, with minimum $\sigma_{\min}$ and maximum $\sigma_{\max}$, and suppose $\sigma_{\max}$ is colorful.
  We claim that $I \cap \Delaunaycol{\chi}$ is, again, an interval.
  To prove this, we show that there is a unique colorful minimum in this interval.
  Let $\sigma_i = \chi^{-1}(i) \cap \sigma_{\max}$ be vertices of $\sigma_{\max}$ with color $i$, for $0\leq i \leq s$. 
  By Lemma~\ref{lem:remove_nontrivial_color}, $\sigma_{\max} \setminus \sigma_i \notin I$ whenever $\card{\sigma_i} > 1$. 
  Therefore, the only colors missing in $\sigma_{\min}$ are the singleton colors. 
  Letting $\sigma' \subseteq \sigma_{\max}$ contain a vertex iff it is the only one of its color in $\sigma_{\max}$, we see that $\sigma_{\min} \cup \sigma'$ is the unique colorful minimum of $I$.

  \smallskip
  It follows that $I \cap \Delaunaycol{\chi}$ is an interval, and the partition of $\Delaunay{\chi}$ into intervals implies a partition of $\Delaunaycol{\chi}$ into intervals,
  so $\alpha$ restricted to $\Delaunaycol{\chi}$ is indeed a generalized discrete Morse function.
\end{proof}

\clearpage
\section{Notation}
\label{app:N}

\begin{table}[h!]
  \centering \vspace{-0.1in}
  \begin{tabular}{ll}
    $[s] = \{0, \ldots, s\}$
      & colors \\
    $\chi \colon A \to [s]$
      & $(s+1)$-coloring of locally finite points \\
    $A_i = \chi^{-1} (i)$
      & points with color $i$ \\
    $\aaa = (a_0, \ldots, a_s)$
      & colorful $s$-simplex \\
    $\Lune{r}{\aaa}$
      & intersection of $s+1$ disks of radius $r$ \\
    $f_{\max} \colon \Rspace^2 \to \Rspace$
      & minmax distance function \\
    \\
    $c_s, c_s^\ppp$
      & asymptotic constants \\
    $\ppp = (p_0, p_1, \ldots, p_s)$
      & probability vector \\
    $\Cost{}{A_0,\ldots,A_s}$
      & cost of lunar EMST \\
    $\Voronoi{A_0,\ldots,A_s}$
      & overlay of Voronoi tessellations \\
    $\Delaunay{A_0,\ldots,A_s}$
      & lunar Delaunay mosaic \\
    $\domain{\aaa}$
      & ($2$-dimensional) domain in tessellation \\
    \\
    $\Delaunay{A}, \Delaunay{A'}$
      & Delaunay mosaic, weighted Delaunay mosaic \\
    $g \colon \Delaunay{A'} \to \Rspace$
      & radius function on weighted Delaunay mosaic \\
    \\
    $Q = [0,1]^2$
      & unit square \\
    $\Cost{s}{A;Q}$
      & average cost of lunar EMSTs \\
    $\Costrestr{s}{A;Q}$
      & average cost of restricted lunar EMSTs \\
    $\Costrooted{s}{A;Q}$
      & average cost of rooted lunar EMSTs \\
    $C_s^{\rm sub}, C_s^{\rm sup}$
      & constants for sub- and super-additivity \\
    $C_s^{\rm one}, C_s^{\rm two}, C_s^{\rm smth}$
      & constants for differences between lunar EMSTs \\
  \end{tabular}
  \caption{Notation used in the paper.}
  \label{tbl:Notation}
\end{table}

\newpage
\section{Results and Definitions}

\begin{itemize}
  \item Section~\ref{sec:1}: Introduction.
  \item Section~\ref{sec:2}: Lunar Generalization of the EMST.
    \begin{itemize}
      \item Lemma~\ref{lem:lunar_alpha_vs_lunar_space} (Lunar Alpha vs Lunar Space).
      \item Lemma~\ref{lem:critical_and_non-critical_cells} (Critical and Non-critical Cells).
      \item Lemma~\ref{lem:counting_critical_cells} (Counting Critical Cells).
      \item Proposition~\ref{prop:chromatic_vs_lunar} (Chromatic vs Lunar).
      \item Corollary~\ref{cor:cost_and_norm} (Cost and Norm).
    \end{itemize}
  \item Section~\ref{sec:3}: Poisson--Delaunay Mosaics.
    \begin{itemize}
      \item Lemma~\ref{lem:radii_and_squared_radii} (Radii and Squared Radii).
      \item Lemma~\ref{lem:crossing_the_boundary} (Crossing the Boundary).
      \item Lemma~\ref{lem:weighted_radii} (Weighted Radii).
      \item Lemma~\ref{lem:chromatic_weighted_radii} (Chromatic Weighted Radii).
    \end{itemize}
  \item Section~\ref{sec:4}: Main Theorem.
    \begin{itemize}
      \item Theorem~\ref{thm:asymptotic_constant} (Asymptotic Constant).
      \item Definition~\ref{dfn:restricted_lunar_EMST} (Restricted Lunar EMST).
      \item Definition~\ref{dfn:rooted_lunar_EMST} (Rooted Lunar EMST).
      \item Corollary~\ref{cor:asymptotic_constant} (Asymptotic Constant).
      \item Lemma~\ref{lem:sub-additivity} (Sub-additivity).
      \item Lemma~\ref{lem:mean_closeness_one} (Mean Closeness One).
      \item Lemma~\ref{lem:super-additivity} (Super-additivity).
      \item Lemma~\ref{lem:cost_of_anchors} (Cost of Anchors).
      \item Lemma~\ref{lem:overlapping_lunes} (Overlapping Lunes).
      \item Lemma~\ref{lem:mean_closeness_two} (Mean Closeness Two).
      \item Lemma~\ref{lem:combinatorial_smoothness} (Combinatorial Smoothness).
    \end{itemize}
  \item Section~\ref{sec:5}: Discussion.
  \\
  \item Appendix~\ref{app:A}: Poisson--Delaunay Mosaics in $\Rspace^d$.
    \begin{itemize}
      \item Theorem~\ref{thm:counting_generalized} (Counting Generalized).
      \item Corollary~\ref{cor:measuring_generalized} (Measuring Generalized).
      \item Lemma~\ref{lem:crossing_the_boundary_details} (Crossing the Boundary, Details).
      \item Theorem~\ref{thm:weighted_counting} (Weighted Counting).
      \item Corollary~\ref{cor:weighted_measuring} (Weighted Measuring).
    \end{itemize}
  \item Appendix~\ref{app:B}: Colorful Simplices of the Chromatic Delaunay Complex.
    \begin{itemize}
      \item Proposition~\ref{prop:colorful_chromatic_vs_lunar} (Colorful Chromatic vs Lunar).
      \item Definition~\ref{dfn:strong_chromatic_genericity} (Strong Chromatic Genericity).
      \item Lemma~\ref{lem:measure-zero_subset} (Measure-zero Subset).
      \item Lemma~\ref{lem:remove_nontrivial_color} (Remove Nontrivial Color).
      \item Theorem~\ref{thm:generalized_discrete_Morse} (Generalized Discrete Morse).
    \end{itemize}
\end{itemize}

\end{document}

%% file: Figs-lunar/weighted.pdf_t
\begin{picture}(0,0)%
\includegraphics{Figs-lunar/weighted.pdf}%
\end{picture}%
\setlength{\unitlength}{3947sp}%
\begingroup\makeatletter\ifx\SetFigFont\undefined%
\gdef\SetFigFont#1#2#3#4#5{%
  \reset@font\fontsize{#1}{#2pt}%
  \fontfamily{#3}\fontseries{#4}\fontshape{#5}%
  \selectfont}%
\fi\endgroup%
\begin{picture}(7266,2744)(1468,-4133)
\put(7501,-3961){\makebox(0,0)[b]{\smash{{\SetFigFont{12}{14.4}{\rmdefault}{\mddefault}{\updefault}{\color[rgb]{0,0,0}$\tt e'$}%
}}}}
\put(2701,-2611){\makebox(0,0)[b]{\smash{{\SetFigFont{12}{14.4}{\rmdefault}{\mddefault}{\updefault}{\color[rgb]{0,0,0}$\tt a$}%
}}}}
\put(4501,-2611){\makebox(0,0)[b]{\smash{{\SetFigFont{12}{14.4}{\rmdefault}{\mddefault}{\updefault}{\color[rgb]{0,0,0}$\tt c$}%
}}}}
\put(5101,-3211){\makebox(0,0)[b]{\smash{{\SetFigFont{12}{14.4}{\rmdefault}{\mddefault}{\updefault}{\color[rgb]{0,0,0}$\tt d$}%
}}}}
\put(5101,-3961){\makebox(0,0)[b]{\smash{{\SetFigFont{12}{14.4}{\rmdefault}{\mddefault}{\updefault}{\color[rgb]{0,0,0}$\tt d'$}%
}}}}
\put(8626,-3961){\makebox(0,0)[b]{\smash{{\SetFigFont{12}{14.4}{\rmdefault}{\mddefault}{\updefault}{\color[rgb]{0,0,0}$\Rspace^1$}%
}}}}
\put(2701,-3961){\makebox(0,0)[b]{\smash{{\SetFigFont{12}{14.4}{\rmdefault}{\mddefault}{\updefault}{\color[rgb]{0,0,0}$\tt a'$}%
}}}}
\put(4426,-3961){\makebox(0,0)[b]{\smash{{\SetFigFont{12}{14.4}{\rmdefault}{\mddefault}{\updefault}{\color[rgb]{0,0,0}$\tt c'$}%
}}}}
\put(3901,-2011){\makebox(0,0)[b]{\smash{{\SetFigFont{12}{14.4}{\rmdefault}{\mddefault}{\updefault}{\color[rgb]{0,0,0}$\tt b$}%
}}}}
\put(7501,-2611){\makebox(0,0)[b]{\smash{{\SetFigFont{12}{14.4}{\rmdefault}{\mddefault}{\updefault}{\color[rgb]{0,0,0}$\tt e$}%
}}}}
\end{picture}%

%% file: Figs-lunar/merging.pdf_t
\begin{picture}(0,0)%
\includegraphics{Figs-lunar/merging.pdf}%
\end{picture}%
\setlength{\unitlength}{3947sp}%
\begingroup\makeatletter\ifx\SetFigFont\undefined%
\gdef\SetFigFont#1#2#3#4#5{%
  \reset@font\fontsize{#1}{#2pt}%
  \fontfamily{#3}\fontseries{#4}\fontshape{#5}%
  \selectfont}%
\fi\endgroup%
\begin{picture}(10994,2444)(1329,-2783)
\put(3601,-2086){\makebox(0,0)[b]{\smash{{\SetFigFont{14}{16.8}{\rmdefault}{\mddefault}{\updefault}{\color[rgb]{0,0,0}$v'$}%
}}}}
\put(9751,-2086){\makebox(0,0)[b]{\smash{{\SetFigFont{14}{16.8}{\rmdefault}{\mddefault}{\updefault}{\color[rgb]{0,0,0}$v$}%
}}}}
\put(4051,-2086){\makebox(0,0)[b]{\smash{{\SetFigFont{14}{16.8}{\rmdefault}{\mddefault}{\updefault}{\color[rgb]{0,0,0}$v''$}%
}}}}
\put(4051,-886){\makebox(0,0)[b]{\smash{{\SetFigFont{14}{16.8}{\rmdefault}{\mddefault}{\updefault}{\color[rgb]{0,0,0}$u''$}%
}}}}
\put(1801,-961){\makebox(0,0)[b]{\smash{{\SetFigFont{14}{16.8}{\rmdefault}{\mddefault}{\updefault}{\color[rgb]{0,0,0}$Q'$}%
}}}}
\put(3601,-886){\makebox(0,0)[b]{\smash{{\SetFigFont{14}{16.8}{\rmdefault}{\mddefault}{\updefault}{\color[rgb]{0,0,0}$u'$}%
}}}}
\put(6076,-961){\makebox(0,0)[b]{\smash{{\SetFigFont{14}{16.8}{\rmdefault}{\mddefault}{\updefault}{\color[rgb]{0,0,0}$Q''$}%
}}}}
\put(9751,-886){\makebox(0,0)[b]{\smash{{\SetFigFont{14}{16.8}{\rmdefault}{\mddefault}{\updefault}{\color[rgb]{0,0,0}$u$}%
}}}}
\end{picture}%

%% file: Figs-lunar/splitting.pdf_t
\begin{picture}(0,0)%
\includegraphics{Figs-lunar/splitting.pdf}%
\end{picture}%
\setlength{\unitlength}{3947sp}%
\begingroup\makeatletter\ifx\SetFigFont\undefined%
\gdef\SetFigFont#1#2#3#4#5{%
  \reset@font\fontsize{#1}{#2pt}%
  \fontfamily{#3}\fontseries{#4}\fontshape{#5}%
  \selectfont}%
\fi\endgroup%
\begin{picture}(10994,2444)(1779,-2783)
\put(9676,-886){\makebox(0,0)[b]{\smash{{\SetFigFont{14}{16.8}{\rmdefault}{\mddefault}{\updefault}{\color[rgb]{0,0,0}$u'$}%
}}}}
\put(3976,-886){\makebox(0,0)[b]{\smash{{\SetFigFont{14}{16.8}{\rmdefault}{\mddefault}{\updefault}{\color[rgb]{0,0,0}$u$}%
}}}}
\put(10126,-886){\makebox(0,0)[b]{\smash{{\SetFigFont{14}{16.8}{\rmdefault}{\mddefault}{\updefault}{\color[rgb]{0,0,0}$u''$}%
}}}}
\put(12226,-961){\makebox(0,0)[b]{\smash{{\SetFigFont{14}{16.8}{\rmdefault}{\mddefault}{\updefault}{\color[rgb]{0,0,0}$Q''$}%
}}}}
\put(7951,-961){\makebox(0,0)[b]{\smash{{\SetFigFont{14}{16.8}{\rmdefault}{\mddefault}{\updefault}{\color[rgb]{0,0,0}$Q'$}%
}}}}
\end{picture}%

%% file: EMST-lunar.bbl
\begin{thebibliography}{21}

\footnotesize{

\bibitem{Aur90}
{\sc F.\ Aurenhammer.}
A new duality result concerning Voronoi diagrams. 
\emph{Discrete Comput.\ Geom.} {\bf 5} (1990), 243--254.

\bibitem{BCDES26}
{\sc R.\ Biswas, S.\ Cultrera di Montesano, O.\ Draganov, H.\ Edelsbrunner and M.\ Saghafian.}
On the size of chromatic Delaunay mosaics. 
\emph{Discrete Comput.\ Geom.} {\bf 75} (2026), 24--47.

\bibitem{BHH59}
{\sc J.\ Beardwood, J.H.\ Halton and J.M.\ Hammersley.}
The shortest path through many points.
\emph{Math.\ Proc.\ Cambridge Phil.\ Soc.} {\bf 55} (1959), 299--327.

\bibitem{CoTr21}
{\sc M.\ Correddu and D.\ Trevisan.}
On minimum spanning trees for random Euclidean bipartite graphs.
\emph{Comb.\ Probab.\ Comput.} {\bf 33} (2024), 319--350.

\bibitem{CDES26}
{\sc S.\ Cultrera di Montesano, O.\ Draganov, H.\ Edelsbrunner, and M.\ Saghafian.} 
Chromatic alpha complexes.
\emph{Found.\ Data Sci.} {\bf 8} (2026), 30--62.

\bibitem{DERS25}
{\sc O.\ Draganov, H.\ Edelsbrunner, S.\ Rosenmeier and M.\ Saghafian.}
Expected length of the Euclidean minimum spanning tree and $1$-norms of chromatic persistence diagrams in the plane.
Manuscript, Institute of Science and Technology Austria, Klosterneuburg, Austria, 2025.
\href{https://arxiv.org/abs/2510.23373}{\texttt{arXiv:2510.23373}}

\bibitem{EdHa10}
{\sc H.\ Edelsbrunner and J.L.\ Harer.}
\emph{Computational Topology. An Introduction.}
Amer.\ Math.\ Soc., Providence, Rhode Island, 2010.

\bibitem{EdNi19}
{\sc H.\ Edelsbrunner and A.\ Nikitenko.}
Weighted Poisson--Delaunay mosaics.
\emph{Theory Probab.\ Appl.} {\bf 64} (2019), 746--770.

\bibitem{ENR17}
{\sc H.\ Edelsbrunner, A.\ Nikitenko and M.\ Reitzner.}
Expected sizes of Poisson--Delaunay mosaics and their discrete Morse functions.
\emph{Adv.\ Appl.\ Probab.} {\bf 49} (2017), 745--767. 

\bibitem{For98}
{\sc R.\ Forman.}
Morse theory for cell complexes.
\emph{Adv.\ Math.} {\bf 134} (1998), 90--145.

\bibitem{Fre09}
{\sc R.\ Freij.}
Equivariant discrete Morse theory.
\emph{Discrete Math.} {\bf 309} (2009), 3821--3829.

\bibitem{GiPo68}
{\sc E.N.\ Gilbert and H.O.\ Pollack.}
Steiner minimal trees.
\emph{SIAM J.\ Appl.\ Math.} {\bf 16} (1968), 1--29.

\bibitem{Kru56}
{\sc J.B.\ Kruskal.}
On the shortest spanning subtree of a graph and the traveling salesman problem.
\emph{Proc.\ Amer.\ Math.\ Soc.} {\bf 7} (1956), 48--50.

\bibitem{Led23}
{\sc M.\ Ledoux.}
Optimal matching of random samples and rates of convergence of empirical measures.
In \emph{Mathematics Going Forward}, eds.: J.-M.\ Morel and B.\ Teissier, Lecture Notes in Mathematics {\bf 2313}, 2023, 615--627.

\bibitem{Lef42}
{\sc S.\ Lefschetz.}
\emph{Algebraic Topology.}
American Mathematical Society Colloquium
Publications, Vol. 27, Amer.\ Math.\ Soc., New York, 1942.

\bibitem{NCBJ24}
{\sc A.\ Natarajan, T.\ Chaplin, A.\ Brown and M.\ Jimenez.}
Morse theory for chromatic Delaunay triangulations.
\href{https://arxiv.org/abs/2405.19303}{\texttt{arXiv:2405.19303}}, 2024.

\bibitem{ScWe08}
{\sc R.\ Schneider and W.\ Weil.}
\emph{Stochastic and Integral Geometry.}
Springer, Berlin, Germany, 2008.

\bibitem{Seg68}
{\sc G.\ Segal.}
Classifying spaces and spectral sequences.
\emph{Math.\ Inst.\ Hautes Études Sci.} {\bf 34} (1968), 105--112.

\bibitem{Ste97}
{\sc J.M.\ Steele.}
\emph{Probability Theory and Combinatorial Optimization.}
CBMS-NSF Regional Conference Series in Applied Mathematics, SIAM, Philadelphia, Pennsylvania, 1997.

\bibitem{Tou80}
{\sc G.\ Toussaint.}
The relative neighborhood graph of a finite planar set.
\emph{Pattern Recognition} {\bf 12} (1980) 261--268. 

\bibitem{Yuk98}
{\sc J.E.\ Yukich.}
\emph{Probability Theory of Classical Euclidean Optimization Problems.}
Springer, Berlin, Germany, 1998.

}

\end{thebibliography}
